\documentclass[11pt]{article}

\usepackage{fullpage}
\usepackage{hyperref}
\usepackage{stmaryrd}
\usepackage{textcomp}
\usepackage{enumerate}
\usepackage{multicol}
\usepackage{algorithm}
\usepackage [noend]{algpseudocode}

\hypersetup{
    unicode=false,          
    colorlinks=true,        
    linkcolor=red,          
    citecolor=green,        
    filecolor=magenta,      
    urlcolor=cyan,           
    hypertexnames=false
}
\usepackage[capitalize]{cleveref}

\usepackage{todonotes}
\usepackage{amsmath}
\usepackage{amssymb}
\usepackage{amsthm}
\newtheorem{theorem}{Theorem}
\newtheorem{lemma}[theorem]{Lemma}
\newtheorem{observation}[theorem]{Observation}
\newtheorem{proposition}[theorem]{Proposition}
\newtheorem{corollary}[theorem]{Corollary}
\newtheorem{assumption}{Assumption}

\usepackage{comment}
\excludecomment{ignore}

\usepackage{epigraph}

\usepackage{tikz}
\usetikzlibrary{shapes}
\newcommand*\routri{\begin{tikzpicture}
   \node at (0,0) [regular polygon, regular polygon sides=3, draw, scale = 0.4, inner sep=1, shape border rotate=180] {Y};
\end{tikzpicture}}
\newcommand{\peace}{\routri}

\newtheorem*{rep@theorem}{\rep@title}
\newcommand{\newreptheorem}[2]{%
\newenvironment{rep#1}[1]{%
\def\rep@title{#2 \ref{##1}}%
\begin{rep@theorem}}%
{\end{rep@theorem}}}
\makeatother

\newreptheorem{theorem}{Theorem}
\newreptheorem{lemma}{Lemma}
\newreptheorem{claim}{Claim}
\newreptheorem{corollary}{Corollary}

\newcommand{\namedref}[2]{\hyperref[#2]{#1~\ref*{#2}}}

\newcommand{\sectionref}[1]{\namedref{Section}{#1}}

\newcommand{\theoremref}[1]{\namedref{Theorem}{#1}}

\newcommand{\figureref}[1]{\namedref{Figure}{#1}}
\newcommand{\algorithmref}[1]{\namedref{Algorithm}{#1}}
\newcommand{\figurerefb}[2]{\hyperref[#1]{Figure~\ref*{#1}#2}}

\newcommand{\lemmaref}[1]{\namedref{Lemma}{#1}}
\newcommand{\lineref}[1]{\namedref{Line}{#1}}

\newcommand{\corollaryref}[1]{\namedref{Corollary}{#1}}
\newcommand{\observationref}[1]{\namedref{Observation}{#1}}
\newcommand{\propositionref}[1]{\namedref{Proposition}{#1}}
\newcommand{\assumptionref}[1]{\namedref{Assumption}{#1}}
\newcommand{\equationref}[1]{\hyperref[#1]{(\ref*{#1})}}

\newcommand{\idlow}[1]{\mathord{\mathcode`\-="702D\tt #1\mathcode`\-="2200}}
\newcommand{\id}[1]{\ensuremath{\idlow{#1}}}

\usepackage{bm}
\newcommand{\real}[1]{\bm{\mathcal{#1}}}
\newcommand{\simulated}[1]{\mathit{#1}}

\usepackage{color}

\begin{document}

\title{
Revisionist Simulations: \\
A New Approach to Proving Space Lower Bounds
}
%
%
%

\author{
	Faith Ellen\\
  \small University of Toronto\\
  \small faith@cs.toronto.edu\\
  \and
Rati Gelashvili\\
  \small University of Toronto\\
  \small gelash@cs.toronto.edu\\
  \and
Leqi Zhu\\
  \small University of Toronto\\
  \small lezhu@cs.toronto.edu\\
}


\maketitle

\begin{abstract}
Determining the number of registers required for solving $x$-obstruction-free (or randomized wait-free) 
  $k$-set agreement for $x \leq k$ is an open problem that highlights important gaps 
  in our understanding of the space complexity of synchronization.
In $x$-obstruction-free protocols, processes are required to return in 
  executions where at most $x$ processes take steps.
The best known upper bound on the number of registers needed to solve this problem among $n>k$ processes is $n-k+x$ registers.
  No general lower bound better than $2$ was known.

We prove that any $x$-obstruction-free protocol solving \mbox{$k$-set}
agreement among $n > k$ processes
must use $\lfloor \frac{n-x}{k+1-x} \rfloor + 1$ or more registers.
Our main tool is a simulation that serves as a reduction from the impossibility 
  of deterministic wait-free $k$-set agreement.
In particular, we show that, if a protocol uses fewer registers,
  then it is possible for $k+1$ processes to simulate the protocol and deterministically solve $k$-set agreement in a wait-free manner, which is impossible.

We generalize this simulation to prove space lower bounds for $x$-obstruction-free protocols solving colorless tasks. In particular, we prove a lower bound of $\lfloor \frac{n}{2} \rfloor +1$ for obstruction-free $\epsilon$-approximate agreement, for sufficiently small $\epsilon$. 
An important aspect of the simulation is the ability of simulating processes to revise 
  the past of simulated processes.
We introduce an \emph{augmented snapshot object}, which facilitates this.

We also prove that any 
lower bound on the number of registers used by obstruction-free protocols 
  applies to protocols that satisfy nondeterministic solo termination.
Hence, our lower bounds for the obstruction-free case
also holds for randomized wait-free protocols. 
In particular, we get a tight lower bound of exactly $n$ registers for solving 
  obstruction-free and randomized wait-free consensus. 

\end{abstract}


\epigraph{\ \\
He who controls the past controls the future. He who controls the present controls the past.}{George Orwell, 1984}

\vspace{-20pt}


\section{Introduction}
The $k$-set agreement problem,
introduced by Chaudhuri~\cite{Cha93}, is a well-known synchronization task in which $n > k$ processes, 
each with an input value, are required to output 
at most different $k$ values, each of which is the input of some process.
This is a generalization of the 
classical consensus problem,
which is the case $k=1$.

Two celebrated results in distributed computing 
are the impossibility of solving consensus deterministically when
at most one process may crash~\cite{FLP85,LA87}
and, more generally,
the impossibility of solving $k$-set agreement deterministically 
when at most $k$ processes may crash~\cite{BG93, HS99, SZ00},
using only registers.
One way to bypass these impossibility results is to 
design protocols that are obstruction-free~\cite{HLM03}.
\emph{Obstruction-freedom} is a termination condition that requires a process to terminate given sufficiently many consecutive steps, i.e.,~from any configuration, if only one process takes steps, then it will eventually terminate. 
$x$-{\em obstruction-freedom}~\cite{Tau17} generalizes this condition:
from any configuration, if only $x$ processes take steps, then they will all eventually terminate.
It is known that $k$-set agreement can be solved using only registers
in an $x$-obstruction-free way for $1 \leq x \leq k$ ~\cite{YNG98}.
Another way to overcome the impossibility of solving consensus 
is to use randomized wait-free protocols, where non-faulty processes are required to 
terminate with probability $1$~\cite{BenOr83}.

It is possible to solve consensus for $n$ processes using $n$ registers
in a randomized wait-free way~\cite{Abr88, AH90, SSW91, AC08} or 
in an obstruction-free way~\cite{GR05, Bow11, Zhu15, BRS15}.
A lower bound of $\Omega(\sqrt{n})$ was proved by Ellen, Herlihy, and Shavit in~\cite{FHS98}.
Recently, Gelashvili proved an $\Omega(n)$ lower bound 
for \emph{anonymous} processes~\cite{Gel15}.
Anonymous processes~\cite{FHS98, AGM02} have no identifiers
and run the same code: all processes with the same input 
start in the same initial state and behave identically 
until they read different values.
Then Zhu proved
that any obstruction-free protocol solving consensus for $n$ processes
requires at least $n-1$ registers~\cite{Zhu16}.
All these lower bounds are actually for protocols that satisfy 
\emph{nondeterministic solo termination}~\cite{FHS98},
which includes both obstruction-free and randomized wait-free protocols.

In contrast, there are big gaps between the best known upper and lower bounds 
on the
number of registers needed for
$k$-set agreement.
The best obstruction-free protocols require $n-k+1$ 
registers~\cite{Zhu15, BRS15}.
Bouzid, Raynal, and Sutra~\cite{BRS15} also give an $x$-obstruction-free protocol
that uses $n-k+x$ registers, 
improving on the $min(n + 2x - k, n)$
space complexity of Delporte-Gallet, Fauconnier, Gafni, 
and Rajsbaum's obstruction-free protocol~\cite{DFGR13}.
All of these algorithms work for anonymous processes.
Delporte-Gallet, Fauconnier, Kuznetsov, and Ruppert~\cite{DFKR15} proved that it
is impossible to solve $k$-set agreement using $1$ register.
For anonymous processes, they also proved a lower bound
of $\sqrt{x(\frac{n}{k}-2)}$ for $x$-obstruction-free protocols,
which still leaves a polynomial gap between the lower and upper bounds. 

There are good reasons why proving lower bounds 
on the number of registers needed for $k$-set agreement may be
difficult.
At a high level, the impossibility results for $k$-set agreement consider
some representation (for example, a simplicial complex)
of all possible process states in all possible executions.
Then, a combinatorial property (Sperner's Lemma~\cite{Lef49book}) is used to prove that, 
roughly speaking, for any given number of steps, there exists an execution leading 
to a configuration in which $k+1$ outputs are still possible.
Although there is ongoing work to develop a more general theory~\cite{GKM14, SHG16, GHKR16},
we do not know enough about the topological representation
of protocols that are $x$-obstruction-free 
or use fewer than $n$ multi-writer registers~\cite{HKR13book}
to adapt topological arguments to prove space lower bounds for $k$-set agreement.
There are similar problems adapting known proofs that do not explicitly use topology~\cite{AC13, AP16}.


Approximate agreement~\cite{DLPSW86} is another important task for which 
no good space lower bound was known.
In $\epsilon$-approximate agreement, each process starts with an input in $\{0, 1\}$.
The processes are required to output values in the interval $[0,1]$ that are all within $\epsilon$ of each other. Moreover, each output value must lie between the smallest input and the largest input.
This problem can be deterministically solved in a \emph{wait-free} manner, i.e. every non-faulty process eventually outputs a value.
The only space lower bound for this problem, $\Omega(\log(\frac{1}{ \epsilon}))$, 
was in a restricted setting with single-bit registers~\cite{Sch96}.
The best upper bounds
are $\lceil\log_2 (\frac{1}{ \epsilon})\rceil$ \cite{Sch96} and $n$ \cite{attiya1994wait}.

\textbf{Our contribution.} 
In this paper, we prove a
lower bound of 
$\lfloor \frac{n-x}{k+1-x} \rfloor + 1$
on the number of registers
necessary for solving $n$-process 
$x$-obstruction-free $k$-set agreement.
As corollaries,
we get a tight lower bound of $n$ registers for obstruction-free consensus
and a tight lower bound of 2 for obstruction-free $(n-1)$-set consensus.
We also prove a space lower bound 
  of $\lfloor \frac{n}{2} \rfloor + 1$ registers 
  for obstruction-free $\epsilon$-approximate agreement,
  for sufficiently small $\epsilon$. More generally, we prove space lower bounds for colorless tasks.
  
In addition, 
  in~\sectionref{sec:ndst2of},
we prove that any lower bound on the number registers needed for obstruction-free protocols 
to solve a task 
 also applies to nondeterministic solo terminating protocols and, in particular, to randomized wait-free protocols
solving that task.
Hence, 
our space lower bounds
for obstruction-free protocols
also apply to such protocols. We also show that the same result may be obtained for a large class of objects.

\textbf{Technical Overview.}
Using a novel simulation, we convert any  obstruction-free protocol for $k$-set agreement 
that uses too few registers to a protocol that solves wait-free $k$-set agreement
using only registers.
Since solving  wait-free $k$-set agreement
is impossible using only registers, this reduction gives a lower bound on the number of registers needed
to solve obstruction-free $k$-set agreement.
This simulation technique, described in detail in~\sectionref{sec:simulation},
  is the main technical contribution of the paper.
It is the first technique that proves lower bounds on space complexity
by applying results obtained by topological arguments.
We also use this new technique to prove a
lower bound on the number of registers needed
  for $\epsilon$-approximate agreement by a reduction from a step complexity 
  lower bound for $\epsilon$-approximate agreement.
Specifically, we convert any obstruction-free protocol for $\epsilon$-approximate agreement 
  to a protocol that uses few registers to a protocol that solves 
  $\epsilon$-approximate agreement for two processes 
  such that both processes take few steps. 

The executions of the simulated processes in our simulation
are reminiscent of the executions
constructed by adversaries in covering arguments~\cite{BL93,AE14}.
In those proofs,
the adversary modifies an execution it has constructed
by revising the past of some process, so that the old and new executions are indistinguishable
to the other processes. It does so
by inserting consecutive steps of the process starting from some 
  carefully chosen configuration.
In our simulation, a real
process may revise the past of a simulated process, 
  in a way that is indistinguishable to other simulated processes.
This is possible because each
simulated process is simulated by a single
real process.
In contrast, in the BG simulation~\cite{BGLR01}, different steps of simulated processes 
  can be performed by different
real processes,
  so this would be
much more difficult to do.
  
A crucial component of our simulation is the use of an {\em augmented snapshot object}, which we implement in a non-blocking manner from registers.
 Like a standard snapshot object, this object consists of a fixed number of components and supports a \id{Scan} operation,
which returns the contents of all components.
However, it generalizes the  \id{update} operation to a \id{Block-Update} operation, which can update multiple components of the object. In addition, a \id{Block-Update} returns 
some information, which is used by our simulation.
The specifications of \id{Block-Update} and our implementation of an
augmented snapshot object appears in Section \ref{sec:fullaugsnapshot}.
\section{Preliminaries}
\label{sec:prelim}
An asynchronous shared memory system consists of a set of processes 
  and instances of base objects, 
which processes use to communicate.
An \emph{object} has a set of possible values and a set of operations, 
  each of which takes some fixed number of inputs and returns a response.
The processes take steps at arbitrary speeds 
  and may fail, at any time, by crashing.
Every step consists of an operation on some base object 
by some process plus local computation
by that process to determine its next state from the response
returned by the operation.

\textbf{Configurations and Executions.}
A \emph{configuration} of a system consists of the state 
  of each process
and the value of each object.
An \emph{initial} configuration
  is determined by the input value of each process.
Each object has the same value in all initial configurations.
A configuration $C$ is \emph{indistinguishable} from a configuration $C'$ to 
  a set of processes $\mathcal{P}$ in the system, if every process in $\mathcal{P}$ 
  is in the same state in $C$ as it is in $C'$ and 
each object in the system has the same value in $C$ as in $C'$. 

A step $e$ by a process $p$ is \emph{applicable} at a configuration $C$ 
  if $e$ can be the next step of process $p$ given its state in $C$. 
If $e$ is applicable at $C$, then we use $Ce$ to denote the configuration 
  resulting from $p$ taking step $e$ at $C$.
A sequence of steps $\alpha = e_1, e_2, \dots$ is \emph{applicable} 
  at a configuration $C$ if $e_1$ is applicable at $C$ and, 
  for each $i \geq 1$, $e_{i+1}$ is applicable at $Ce_1\cdots e_i$. 
In this case, $\alpha$ is called an \emph{execution from} $C$. 
An execution $\alpha \beta$ denotes the execution $\alpha$
  followed by the execution $\beta$.
A configuration $C$ is \emph{reachable} if there exists a finite execution 
  from an initial configuration that results in $C$.

For a finite execution $\alpha$ from a configuration $C$, 
  we use $C\alpha$ to denote the configuration reached after applying $\alpha$ to $C$. 
If $\alpha$ is empty, then $C\alpha = C$. 
We say an execution $\alpha$ is \emph{$\mathcal{P}$-only}, 
  for a set of processes $\mathcal{P}$, 
  if all steps in $\alpha$ are by processes in $\mathcal{P}$. 
A $\{p\}$-only execution, for some process $p$, 
  is also called a \emph{solo execution by $p$}. 
Note, if configurations $C$ and $C'$ are indistinguishable 
  to a set of processes $\mathcal{P}$, 
  then any $\mathcal{P}$-only execution from $C$ is applicable at $C'$.

\textbf{Implementations and Linearizability.}  
An \emph{implementation} of an object specifies, for each process and each operation of the object, a deterministic procedure describing how the process carries out the operation. The \emph{execution interval} of an invocation of an operation in an execution is the subsequence of the execution that begins with its first step and ends with its last step.
If an operation does not complete, for example, if the process that invoked it crashed
before receiving a response, then its execution interval is infinite.
An implementation of an object is \emph{linearizable} if, for every execution, there is a point in
each operation's execution interval,
called the \emph{linearization point} of the operation, such that the operation can be said to have taken place atomically at that point \cite{HW90}. This is equivalent to saying that the operations can be ordered (and all incomplete operations can be given responses)
so that any operation which ends before another one begins is ordered earlier
and the responses of the operations are consistent with the sequential specifications of the object~\cite{HW90}. 


\textbf{Progress Conditions.} An implementation of an object is \emph{wait-free} if every process is able to complete its current operation on the object after taking sufficiently many steps, regardless of what other processes are doing. An implementation is \emph{non-blocking} if infinitely many operations are completed in every infinite execution. 

A protocol is \emph{$x$-obstruction-free} if, from any configuration $C$ and for any subset $P$ of at most $x$ processes, every process in $P$ that takes sufficiently many steps after $C$ outputs a value, as long as only processes in $P$ take steps after $C$. A protocol is 
\emph{obstruction-free} if it is 1-obstruction-free  and 
\emph{wait-free} if it is $n$-obstruction-free.

\textbf{Registers and Snapshot objects.}
A \emph{register} is an object that supports two operations, $\id{write}$ and $\id{read}$.
A $\id{write(v)}$ operation writes value $v$ to the register,
  and a $\id{read}$ operation returns the last value that was written to the 
  register before the read.
A \emph{multi-writer register} allows all processes to write to it,
  while a \emph{single-writer register} can only be written to by one fixed process.
A process is said to be \emph{covering} a register if its next step
  is a write to this register. 
A \emph{block write} is a consecutive sequence of \id{write} operations to different 
  registers performed by different processes.

An \emph{$m$-component multi-writer
snapshot} object~\cite{AADGMS93}
stores a sequence of $m$ values and
supports two operations, $\id{update}$ and $\id{scan}$. An $\id{update}(j,v)$ operation
sets component $j$ of the object to $v$. 
A $\id{scan}$ operation returns the current \emph{view}, 
consisting of 
the values of all components.
A \emph{single-writer} snapshot object 
shared by a set of processes has one component for each process
and each process may only \id{update} its own component.
A process is said to be covering component $j$ of a snapshot object if its next step
  is an update to the component $j$.
A \emph{block update} is a consecutive sequence of \id{update} operations to different 
  components of a snapshot object performed by different processes.

It is easy to implement $m$ registers from an $m$-component multi-writer snapshot object,
by  replacing each \id{write} to the $j$'th register by an \id{update} to the $j$'th component
and replacing a \id{read} to the $j$'th register by a \id{scan} and then discarding 
all but the value of the $j$'th component.
An $m$-component snapshot object can also
be implemented from $m$ registers~\cite{AADGMS93}.
 
\textbf{Tasks and Protocols.} 
A \emph{task} specifies a set of allowable combinations of inputs to the processes and, for each such combination, what combinations of outputs can be returned by the processes.
A \emph{protocol} for a task provides a procedure for each process to compute its
output, so that the task's specifications are satisfied. 

A task is \emph{colorless} if the input or output of any process may be the input or output, respectively, of another process. Moreover, the specification of the task does not depend on the number of processes in the system. More precisely, a colorless task is a triple $(\mathcal{I},\mathcal{O},\Delta)$, where $\mathcal{I}$ contains sets of possible inputs, $\mathcal{O}$ contains sets of possible outputs, and, for each input set $I \in \mathcal{I}$, $\Delta(I)$ specifies a subset of $\mathcal{O}$, corresponding to valid outputs for $I$. Moreover, $\mathcal{I}$, $\mathcal{O}$, and $\Delta(I)$, for each $I \in \mathcal{I}$, are closed under taking subsets; i.e.~if a set is present, then so are its non-empty subsets. The following are all examples of colorless tasks:

\begin{itemize}
	\item \emph{Consensus}: Each process begins with an arbitrary value as its input  and, if it does not crash, must output a value such that no two processes output different values and each output value is the input of some process.
	\item \emph{$k$-Set agreement}: Each process begins with an arbitrary value as its input  and, if it does not crash, must output a value such that at most $k$ values are output and each output value is the input of some process.
	\item \emph{$\epsilon$-Approximate Agreement}: Each process begins with an arbitrary (real) value as its input and, if it does not crash, must output a value such that any two output values are at most $\epsilon$ apart. Moreover, the set of output values is in the interval $[min,max]$, where $min$ and $max$ are the smallest and largest input values, respectively.
\end{itemize} 

The \emph{space complexity} of a protocol 
is the maximum number of registers used in any execution of the protocol.
Each $m$-component snapshot object it uses counts as $m$ registers.
The \emph{space complexity} of a task is the minimum space complexity of any protocol
for the task.


\subsection{Our Setting}
We consider two asynchronous shared memory systems, 
the simulated system and the real system.

%
%

\textbf{Simulated system.}
The \emph{simulated system} consists of $n$ simulated processes, $p_1, \ldots, p_n$,
that communicate through an $m$-component multi-writer snapshot object.
Thus, any task that can be solved in the simulated system has space complexity at most $m$.

Without loss of generality, we assume that each process $p_i$ alternately performs \id{scan} and \id{update} 
  operations on the snapshot object:
Between two consecutive \id{update} operations, $p_i$ can perform a \id{scan}
  and ignore its result.
If $p_i$ is supposed to perform multiple consecutive \id{scan}s,
  it can, instead, perform one \id{scan} and use its result as the result
  of the others.
This is because it is possible for all these \id{scan}s to occur
  consecutively in an execution, in which case, they {\em would}
all  get the same result.


\textbf{Real system.}
The real system consists of $f$ real processes, $q_1, \ldots, q_f$, that communicate through
a single-writer 
snapshot object.
For clarity of presentation, real processes
use single-writer registers in addition to the single-writer snapshot object.
The single-writer registers to which a particular process writes can be treated
as additional separate fields of the component of the snapshot object
belonging to that process.
In~\sectionref{sec:fullaugsnapshot}, we define and implement an
  \emph{$m$-component augmented snapshot object}
  shared by the real processes.

In our simulation,
the processes in the simulated system are partitioned into $f$
  sets, $P_1, \ldots , P_f$,
and real process $q_i$
is solely responsible for simulating
the actions of
all processes in $P_i$ in the simulated system. 
This is illustrated in~\figureref{fig:overview}.

\begin{figure*} [!ht]
	\centering
	\resizebox{0.8\textwidth}{!}{
	\includegraphics[scale=1]{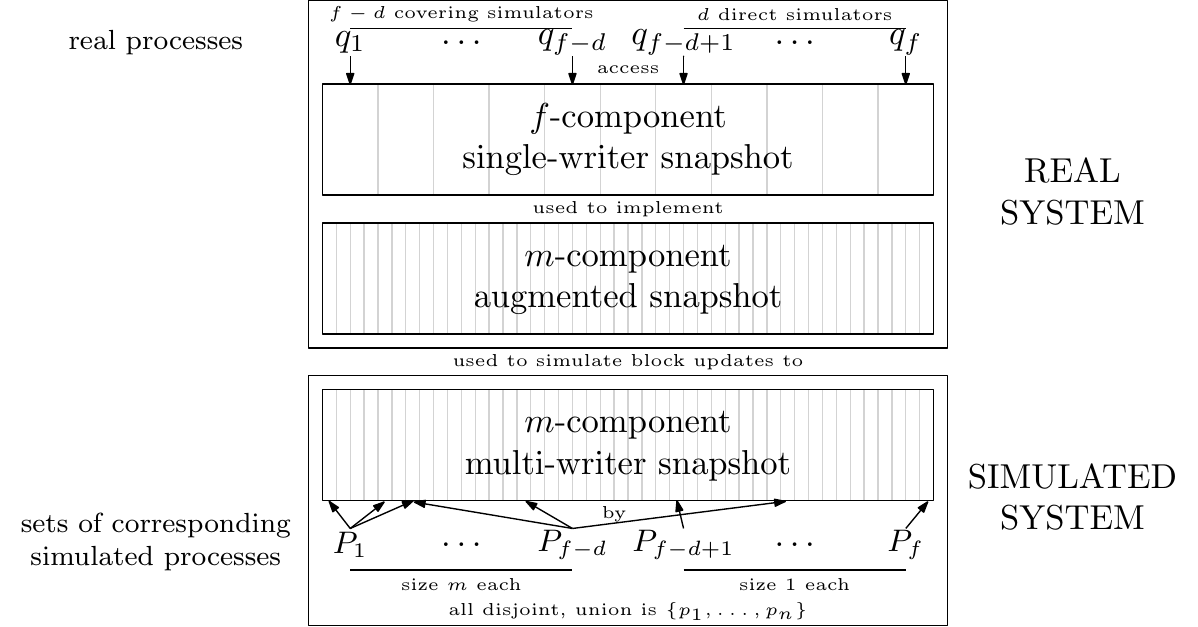}}
	\caption{High-level overview of real and simulated systems.}
	\label{fig:overview}
\end{figure*}

\section{Augmented Snapshot Object}
\label{sec:fullaugsnapshot}
In this section, we define an 
augmented snapshot object 
  and show how it can be deterministically implemented
in the real system.
This object plays a central role in our simulation.
It is used by real processes to simulate steps performed by simulated processes.
In particular, a real process $q_i$ uses this object to simulate 
  an $\id{update}$ or a $\id{scan}$ by any simulated process in $P_i$, 
  or a block update by
any subset of
  processes in $P_i$.
Our simulation, 
which is explained in~\sectionref{sec:simulation}, 
is 
non-standard.
Unfortunately, to satisfy its technical requirements, 
the augmented snapshot
has to satisfy
some
non-standard properties.

An $m$-component
augmented snapshot object is
a generalization of an $m$-component multi-writer snapshot object. 
A \id{Scan} operation returns the current \emph{view}, 
consisting of  the values of all components.
The components can be updated using a \id{Block-Update} operation. 
The key difference between a multi-writer snapshot and an augmented snapshot is that 
a \id{Block-Update} may update \emph{multiple} components, 
 although not necessarily atomically.
In addition, a \id{Block-Update} may return a view from some earlier point in the execution.
Otherwise, it returns a special \emph{yield} symbol, $\peace$.

A linearizable, non-blocking implementation of an augmented snapshot object
in the real system is impossible.
This is because a
\id{Block-Update} operation that updates 2 components
would then be the same as a \id{2-assignment} operation.
However,  \id{2-assignment}, together with \id{read} or \id{Scan},
can be used to deterministically solve wait-free consensus among 2 processes~\cite{her91},
which is impossible in the real system~\cite{AADGMS93,LA87}.

Instead, a \id{Block-Update} operation can be considered to be
a sequence of atomic \id{Update} operations,
which each update one component of 
 the augmented snapshot object.
(Analogously, a \emph{collect} operation \cite{beame1986limits,attiya1994wait} 
 is not atomic, but the individual reads that comprise it are atomic.) 
  
\subsection{Specification}
An \emph{$m$-component augmented snapshot object}, $\mathcal{M}[1..m]$, shared by $f$ processes, 
$q_1,\dots,q_f$, consists of $m$ components and supports two operations, $\id{Scan}$ and $\id{Block-Update}$, which can be performed by all processes. 
A \emph{view} of the augmented snapshot consists of the value of each of the $m$ components
at some point in an execution.
A \id{Scan} operation returns the current view.
A \id{Block-Update} operation to a sequence of $r \geq 1$ different components $[j_1,\dots,j_r]$ of $\mathcal{M}$ with 
a sequence of $r$ values $[v_1,\dots,v_r]$ is comprised of a sequence of \id{Update} operations $U_1,\dots,U_r$. Each \id{Update} $U_g$ atomically sets $\mathcal{M}[j_g]$ to $v_g$.
These \id{Update} operations may occur in any order.
A \id{Block-Update} operation also returns either $\peace$ or a view of $\mathcal{M}$.

A \id{Block-Update} that does not return $\peace{}$ is called \emph{atomic}.
We require that every execution has a linearization in which the \id{Update} operations comprising each atomic \id{Block-Update} are linearized consecutively.

Consider any atomic \id{Block-Update}, $B$.
Let $Z$ be the first 
\id{Update} in $B$.
Let $Z'$ be the
last \id{Update} prior to $Z$ that is part of
an atomic \id{Block-Update} or the beginning of the execution, if there is no such \id{Update}.
Then $B$ must return a view of $\mathcal{M}$  at some point $T$ between $Z'$ and $Z$
such that
no \id{Scan} occurs
between $T$ and $Z$.

We  prove that, in our implementation, a \id{Block-Update} only returns $\peace{}$
under certain circumstances, as described in Theorem~\ref{thm:fulnb}.
For example,
a \id{Block-Update} by process $q_1$ will always be atomic
and, if a \id{Block-Update} experiences no step contention~\cite{AGK05},
it will be atomic.
The simulation in Section \ref{sec:simulation} relies on this property of our implementation.

\subsection{Implementation}
In this section, we describe how to implement an $m$-component augmented snapshot object, $\mathcal{M}$, shared by $f$ processes, $q_1,\dots,q_f$, in the real system. Our implementation is non-blocking: every  \id{Block-Update} operation is wait-free, while a \id{Scan} operation can only be blocked by an
infinite sequence of concurrent \id{Block-Update} operations.

The implementation uses a shared single-writer
snapshot object $\mathsf{H}[1..f]$.
All $f$ components of $\mathsf{H}$ are initially $\bot$.
The $i$'th component of $\mathsf{H}$ is used by $q_i$ to
record a list of every \id{Update} it performs,
each represented by a \emph{triple}.   A
triple contains a component of $\mathcal{M}$, 
a value, and a \emph{timestamp}.
Each time $q_i$ performs a \id{Block-Update} to $r$ components of $\mathcal{M}$, it appends $r$ triples to $\mathsf{H}[i]$, all with the same timestamp.
For clarity of presentation,
we also use $n(n-1)$ unbounded arrays of single-writer registers, 
$\mathsf{L}_{i,j}$, for all $j \neq i$, each indexed by the non-negative integers.
Each register is initially $\bot$.
Only process  $q_i$ can write to $\mathsf{L}_{i,j}[b]$ and only process $q_j$ reads from it. The register $\mathsf{L}_{i,j}[b]$ is used by $q_i$ to help $q_j$ determine what value to return from its $b$'th \id{Block-Update}.

The set of arrays $\mathsf{L}_{i,j}$, for all $j \neq i$, may be viewed as an additional field, $\mathsf{L}$, of $q_i$'s component $\mathsf{H}[i]$. To write value $v$ to $\mathsf{L}_{i,j}[b]$, $q_i$ updates $\mathsf{H}[i]$, appending $(j,b,v)$ to the $\mathsf{L}$ field. If a process $q_j$, $j \neq i$, wishes to read the value of $\mathsf{L}_{i,j}[b]$, then it scans $\mathsf{H}$ and checks if a triple $(j,b,v)$ is present in the $\mathsf{L}$ field of $\mathsf{H}[i]$. If not, it considers the value of $\mathsf{L}_{i,j}[b]$ to be $\bot$. Otherwise, $q_j$ finds the last such triple $(j,b,v)$ and it considers the value of $\mathsf{L}_{i,j}[b]$ to be $v$. 

Observe that, given this representation of $\mathsf{L}_{i,j}$, it is possible for $q_i$ to perform a sequence of writes to $\mathsf{L}_{i,j}$, for $j \neq i$, by performing a single update to $\mathsf{H}$. Similarly, it can read the arrays $\mathsf{L}_{j,i}$, for all $j \neq i$, by performing a single scan on $\mathsf{H}$.

\textbf{Notation.}
We use upper case letters to denote instances
of \id{scan} and \id{update} on $\mathsf{H}$, instances of \id{read} and \id{write} on single-writer registers,
and instances of \id{Scan}, \id{Block-Update} and \id{Update} on $\mathcal{M}$.
The corresponding lower case letter denotes the result of a \id{scan}, \id{read}, or
\id{Block-Update} operation.
For example, $h$ denotes the result of a \id{scan} $H$.
We use $h_i$ to denote the value of the $i$'th component of $h$ and
$\# h_i$ to denote the number of \id{Block-Update} operations $q_i$ has performed on $\mathcal{M}$, which is exactly the number of different timestamps associated with the triples recorded in $h_i$. The only shared variables are $\mathsf{H}$ and $\mathsf{L}_{i,j}$, the rest are local variables.

\textbf{Auxiliary Procedures.}
A \emph{timestamp} is a label from a partially ordered set, which can be associated with an operation, such that, if one operation completes before another operation begins,
the first operation has a smaller timestamp~\cite{Lam78}.
We use a variant of vector timestamps~\cite{Fid91, Mat89,  AW04book}:  Each
\emph{timestamp} is an $f$-component vector of non-negative integers, with one component per process. Timestamps are ordered lexicographically. 
We use $t' \succ t$ to denote that timestamp $t'$ is lexicographically 
larger than timestamp $t$ and $t' \succeq t$ to denote that 
$t'$ is lexicographically at least is large as $t$.

Let $h$ be the result of a  \id{scan}.
Process $q_i$ \emph{generates a new timestamp} 
$t = (t_1,\dots,t_f)$ from $h$ using the locally computed function $\textsc{New-Timestamp}(h)$.
It sets $t_j$ to $\# h_j$ for all $j \neq i$ and sets $t_i$ to $\# h_i + 1$.

\begin{algorithm} [H]
	\begin{algorithmic} [1]
\Procedure{New-Timestamp}{$h$} 
		\For{$j \in \{1,\dots,f\}$} 
		\State $t_j \gets \# h_j$
		\EndFor
		\State \Return $(t_1,\dots,t_{i-1},t_i+1,t_{i+1},\dots,t_f)$	
		\EndProcedure
	\end{algorithmic}
	\caption{Generating a new timestamp from the result $h$ of an $\mathsf{H}.\id{scan}$ for process $q_i$.}
	\label{alg:newtimestamp}
\end{algorithm}

For each 
$j \in \{1,\ldots,m\}$,
let $v_j$ be the value with the lexicographically largest associated timestamp 
among all update triples $(j, v, t)$ in all components of $h$, or $\bot$ if no such triple exists.
The \emph{view} of $h$, denoted $\id{view}(h)$, is the vector $(v_1,\dots,v_m)$. 
It is obtained using the locally computed function $\textsc{Get-View}(h)$.

\begin{algorithm} [H]
	\begin{algorithmic}[1]
		\Procedure{Get-View}{$h$}
		\For{$j \in \{1,\dots,m\}$}
		\If{there is an update triple in $h$ with first component $j$}
		\State $t \gets \max\{t' : (j,v',t')$ is a triple in $h\}$
		\State let $(j,v,t)$ be the \emph{unique} triple in $h$ with component $j$ and timestamp $t$
		\State $v_j \gets v$
		\Else
		\State $v_i \gets \bot$
		\EndIf
		\EndFor
		\State \Return $(v_1,\ldots,v_m)$
		\EndProcedure
	\end{algorithmic}
	\caption{Generating a view of the object given the result $h$ of an $\mathsf{H}.\id{scan}$ for process $q_i$.}
	\label{alg:get-view}
\end{algorithm}

\textbf{Main Procedures.}
To perform a $\id{Scan}()$ of $\mathcal{M}$, process $q_i$ repeatedly performs 
\id{scans} of $\mathsf{H}$
until two consecutive results are the same. Then $q_i$
returns  the \id{view} of its last \id{scan}. 
Notice that \id{Scan} is not necessarily wait-free.
However, it can only be blocked by an infinite sequence of \id{Block-Update} operations
that modify $\mathsf{H}$ between every two \id{scan} operations performed by the \id{Scan}.
To help other processes determine what
to return from a \id{Block-Update},
$q_i$ records the result, $h$, of each  \id{scan}
in register $\mathsf{L}_{i,j}[\# h_j]$, for all $j \neq i$.

\begin{algorithm} [!ht]
	\begin{algorithmic} [1]
		\Procedure{Scan}{}
		\State $h' \gets \mathsf{H}.\id{scan}()$
		\Repeat
		\State $h \gets h'$
		\For {$j \in \{1,\dots,f\} - \{i\}$} \label{line:xloop0}
		\State $\mathsf{L}_{i,j}[\#h_j].\id{write}(h)$\label{line:xwriteL}
		\EndFor
		\State $h' \gets \mathsf{H}.\id{scan}()$\label{line:xlastscan}
		\Until{$h = h'$ \label{line:xscancond}}
		\State \Return $\Call{Get-view}{h}$
		\EndProcedure 
	\end{algorithmic}
	\caption{Implementation of $\id{Scan}$ for process $q_i$.}
	\label{alg:augscan}
\end{algorithm}

To perform a $\id{Block-Update}([j_1,\ldots,j_c],[v_1,\ldots,v_c])$ of $\mathcal{M}$, 
$q_i$ first performs a scan $H$ of $\mathsf{H}$. 
Then  it generates a timestamp, $t$, from the result, $h$, of $H$
and appends the triples $(j_1,v_1,t)$, $\ldots$, $(j_c,v_c,t)$ to $\mathsf{H}[i]$
via an \id{update}.
This associates the same timestamp, $t$, with the \id{Block-Update} and
each of the \id{Update} operations comprising it.
Next, process $q_i$  helps processes with lower identifiers by performing another \id{scan} $G$ of $\mathsf{H}$
and recording its result, $g$, in $L_{i,j}[\# g_j]$ for all $j < i$.
Then $q_i$ performs a third \id{scan} to check whether any process 
  with a higher identifier has performed an $\id{update}$ after $H$.
If so, $q_i$ returns \peace{}.
This is the only way in which a \id{Block-Update} can return \peace{}.
Consequently, 
all \id{Block-Update} operations performed by $q_1$
  are atomic.

If $q_i$ does not return \peace{}, it
reads $\mathsf{L}_{j,i}[b]$  for all $j \neq i$,
where $b$ is the number of \id{Block-Update} operations that $q_i$ 
had previously performed.
It determines which among these and $h$
is the result of the latest \id{scan} and returns its view.
The mechanism for determining the latest \id{scan} is described next.

\begin{algorithm} [!ht]
	\begin{algorithmic} [1]
		\Procedure{Block-Update}{$[j_1,\dots,j_r],[v_1,\dots,v_r]$}
		\State $h \gets \mathsf{H}.\id{scan}()$\label{line:xh}
		\State $t \gets \Call{New-timestamp}{h}$\label{line:xts}
		\State $\mathsf{H}.\id{update}_i(h_i \oplus [(j_1,v_1,t), \ldots, (j_r, v_r, t)])$ \label{line:xupdate}
		\State $g \gets \mathsf{H}.\id{scan}()$\label{line:xhelp}
		\For {$j \in \{1,\dots,i-1\}$} \label{line:xloop1}
		\State $\mathsf{L}_{i,j}[\#g_j].\id{write}(g)$\label{line:xlow}
		\EndFor 
		\State $h' \gets \mathsf{H}.\id{scan}()$\label{line:xhh}
		\If {$h'$ contains new $\id{Block-Update}$ \label{line:xheqhh}}
		\State \Return $\peace$
		\EndIf
		\State $\mathit{last} \gets h$ \label{line:xlastinit}
		\For{$j \in \{1,\dots,f\} - \{i\}$} \label{line:xloop2}
		\State $r[j] \gets \mathsf{L}_{j,i}[\# h_i].\id{read}()$\label{line:xreadLj}
		\If{$r[j] \neq \bot$ {\rm and} $\mathit{last}$ {\rm is a proper prefix of} $r[j]$\label{line:xlastmax}}
		\State $\mathit{last} \gets r[j]$\label{line:xlastmax2}
		\EndIf
		\EndFor 
		\State \Return $\Call{Get-View}{\mathit{last}}$\label{line:xbview}
		\EndProcedure
	\end{algorithmic}
	\caption{Implementation of $\id{Block-Update}$ for process $q_i$.}
	\label{alg:augblock}
\end{algorithm}

If $h$ and $h'$ are the results of two scans of $\mathsf{H}$
and, if $h_j$ is a prefix of $h_j'$ for all $j \in \{1,\ldots,f\}$, then
we say that $h$ is a \emph{prefix} of $h'$. In addition, if $h_j \neq h'_j$
for some $j \in \{1,\ldots,f\}$, we say that $h$ is a \emph{proper} prefix of $h'$.
Since each \id{update} to the single-writer snapshot 
$\mathsf{H}$ appends one or more update triples to a component,
the following is true.

\begin{observation}
\label{obs:prefix-obs}
Let $H$ and $H'$ be {\em scans} of $\mathsf{H}$ with results $h$ and $h'$, respectively. If $H$ occurred before $H'$, then $h$ is a prefix of $h'$. Conversely, if $h$ is a proper prefix of $h'$,
then $H$ occurred before $H'$.
\end{observation}

Thus, by~\observationref{obs:prefix-obs}, for any set of \id{scans},
the result of the earliest of these \id{scans}
is a prefix of the result of every other \id{scan} in the set.

The next lemma shows that our implementation of $\id{Block-Update}$ is wait-free while our implementation of $\id{Scan}$ is non-blocking. 

\begin{lemma} [Step Complexity]
	\label{lem:impldetails}
	Each $\id{Block-Update}$ operation consists of 6 steps. If $k$ is the number of different updates by other processes (which append update triples) that are concurrent with an $\id{Scan}$ operation, then it completes after at most $2k+3$ steps.
\end{lemma}
\begin{proof}
	
	The writes in the loop on lines~\ref{line:xloop1}--\ref{line:xlow} in the pseudocode for an $\id{Block-Update}$ may be simultaneously performed by a single update. Similarly, the reads in the loop on lines~\ref{line:xloop2}--\ref{line:xlastmax2} may be simultaneously performed with a single scan. Thus, each $\id{Block-Update}$ operation consists of 6 operations on the single-writer snapshot $\mathsf{H}$.
	
	An $\id{Scan}$ operation begins with a scan of $\mathsf{H}$. The writes in the loop on lines~\ref{line:xloop0}--\ref{line:xwriteL} may be simultaneously performed by a single update. Hence, each iteration of the loop performs two steps: an update and a scan. Each unsuccessful iteration of the loop is caused by a different update by another process  that occurs between the scan in that iteration and the scan in the previous iteration. In addition,
	there is one successful iteration of the loop.
	Hence, the $\id{Scan}$ performs at most $2k+3$ steps.
\end{proof}

\subsection{Proof of Correctness}
\label{sec:augsnap}

In this section, we prove that our implementation is correct. We begin by describing the linearization points of our operations.

\textbf{Linearization Points.} 
A complete \id{Scan} operation is linearized at its last \id{scan} of $\mathsf{H}$,
performed on~\lineref{line:xlastscan}.
Now consider a  \id{Block-Update}, 
with associated timestamp $t$, that updates components $j_1,\ldots,j_r$.
For $1 \leq g \leq r$, the \id{Update} to component $j_g$
is linearized at the first point in the execution at which
$\mathsf{H}$ contains a triple beginning with
$j_g$ and ending with a timestamp $t' \succeq t$. 
If multiple \id{Update} operations are linearized at the same point, then they are
ordered by their associated timestamps (from earliest to latest)
and then in increasing order of the components they update.

Each \id{Update}  of a \id{Block-Update}, $B$, performed without step contention is linearized
at $B$'s  \id{update}  to $\mathsf{H}$ on~\lineref{line:xupdate}. 
However, it is not possible to do this for all \id{Block-Update} operations;
otherwise, we would be implementing
a linearizable, non-blocking augmented snapshot, which, as discussed
earlier, is impossible.
In our linearization, 
if an \id{Update}, $U$, that is part of $B$ updates a component which is also
updated by an \id{Update}, $U'$,  that is part of a concurrent
\id{Block-Update} by a process with a lower identifier, then $U$ may be linearized before $U'$.

We now prove 
a useful property of our helping mechanism.
\begin{lemma}
	\label{lem:xhelp}
	Let $Y$ be a \id{write} by process $q_j$, where 
	it writes $h$ to $\mathsf{L}_{j,i}[\# h_i]$.
	Let $B$ be a \id{Block-Update} in which process $q_i$ \id{read}s $r[j]$ 
	from $\mathsf{L}_{j,i}[\# h_i]$ on~\lineref{line:xreadLj} after $Y$.
	Let $\ell$ be the value of $\mathit{last}$ when $B$ returns on~\lineref{line:xbview}.
	Then, $h$ is prefix of $\ell$.
\end{lemma}
\begin{proof}
	By~\observationref{obs:prefix-obs},~\lineref{line:xlastmax} and~\lineref{line:xlastmax2}, 
	$r[j]$ is a prefix of $\ell$.
	Hence, it suffices to show that $h$ is a prefix of $r[j]$.
	Suppose $r[j] \neq h$. Then $r[j]$ was written by $q_j$
	after $Y$ and it was the result of a \id{scan} by $q_j$
	that occurs after the \id{scan} by $q_j$ that returns $h$.
	It follows by~\observationref{obs:prefix-obs} that $h$ is a prefix of $r[j]$.
\end{proof}
Since $q_i$ only appends new triples on~\lineref{line:xupdate}, we also have the following.
\begin{observation}
	\label{obs:xbeforex}
	Let $X$ be the first \id{update} performed on~\lineref{line:xupdate} by process $q_i$
	after some \id{scan} $H$ of $\mathsf{H}$ with result $h$.
	Let $G$ be any other \id{scan} of $\mathsf{H}$ before $X$ with result $g$.
	Then, $\# g_i \leq \# h_i$.
\end{observation}
Our linearization rule for $\id{Update}$s implies the following observations.
\begin{observation}
	\label{obs:xupdlin}
	Let $U$ be an \id{Update} to component $j$ with an associated timestamp $t$ that is part of a \id{Block-Update}
	and let $X$ be any \id{update} to $\mathsf{H}$ that appends an update triple 
	with component $j$ and timestamp $t' \succeq t$ to $\mathsf{H}$.
	Then $U$ is linearized no later than $X$. 
\end{observation}
\begin{observation}
	\label{obs:xtscontain}
	If a \id{scan} $H$ of $\mathsf{H}$ occurs after the linearization point of an \id{Update} $U$ 
	to component $j$ with associated timestamp $t$, then the result of $H$ contains 
	an update triple with component $j$ and timestamp at least as large as $t$.
\end{observation}

We say that  the result, $h$, of a \id{scan} of $\mathsf{H}$ \emph{contains} a timestamp $t$,
if $h$ (or, more precisely, some component $h_i$ of $h$) contains an update triple with timestamp $t$. 
The corollary of the next lemma says that a timestamp generated from $h$
is lexicographically larger than any timestamp contained in $h$.

\begin{lemma}
	\label{lem:tsdom}
	For any timestamp $t$ contained in the result, $h$, of a \id{scan} $H$ of $\mathsf{H}$, 
	$\# h_j \geq t_j$, for all $1 \leq j \leq f$. 
\end{lemma}
\begin{proof}
	Suppose $t$ is generated from the result $h'$ of a scan $H'$ by some process $q_i$.
	Then $t_{i} = \# h'_{i} + 1$ and $t_j = \# h'_j$, for $j \neq i$.
	Since $q_i$ appends an update triple with timestamp $t$ to $\mathsf{H}[i]$
	before $t$ is contained in the result of a \id{scan}, 
	$\# h_i \geq \# h'_i +1 = t_i$ and
	$H$ occurs after $H'$.
	By \observationref{obs:prefix-obs}, $h'$ is a prefix of $h$.
	Hence, $\#h_j \geq \# h'_j = t_j$, for $j \neq i$.
\end{proof}
\begin{corollary}
	\label{cor:tsdom}
	Let $h$ be the result of a \id{scan} and let $t = \textsc{New-Timestamp}(h)$
	by any process.
	Then, for any timestamp $t'$ contained in $h$, $t' \prec t$.
\end{corollary}
Now we show that timestamps are unique.
\begin{lemma}
	\label{lem:tsunique}
	Any two triples appended to $\mathsf{H}$ that involve the same component
	of $\mathcal{M}$ are associated with a different timestamp.
\end{lemma}
\begin{proof}
	We show that every \id{Block-Update} operation is associated 
	with a different timestamp.
	Since no \id{Block-Update} operation appends more than one triple 
	for any component of $\mathcal{M}$, the claim follows. 
	
	Suppose two processes $q_i \neq q_{j}$ generate timestamps $t$ and $t'$ from \id{scans} 
	$H$ and $H'$ of $\mathsf{H}$ that return $h$ and $h'$, respectively. 
	Then $t_i = \# h_i + 1$, $t_j = \# h_j$,  $t'_{j} = \# h'_{j} + 1$, and $t'_i = \# h'_i$. 
	If $t = t'$, then $\# h_i + 1 = \# h_i'$ and $\# h_j' + 1 = \# h_j$. 
	It follows that $\# h_i < \# h'_i$ and $\# h_j > \# h'_j$. 
	However, by~\observationref{obs:prefix-obs}, this is impossible. Therefore, $t \neq t'$.
	
	Now, consider two timestamps generated by the same process $q_i$.
	Since $q_i$ appends one or more updates triples with timestamp $t$ to $\mathsf{H}[i]$ 
	immediately after it generates $t$, the result of any subsequent scan by $q_i$ contains $t$.
	Thus, by~\corollaryref{cor:tsdom}, 
	any timestamp $t'$ generated by $q_i$ after $t$ is lexicograpically larger than $t$.
\end{proof}

Next, we show that, 
\id{Block-Update} operations that do not return \peace{} can be considered to 
take effect atomically at their update on~\lineref{line:xupdate}.
\begin{lemma}
	\label{lem:xbots}
	Let $B$ be a \id{Block-Update} operation performed by $q_i$ that does not return \peace{}.
	Let $H$ and $G$ be the \id{scan} operations on~\lineref{line:xh} and~\lineref{line:xhh} 
	in $B$, respectively.
	Then, no \id{update} may be performed by $q_j$ with $j < i$ between $H$ and $G$.  
\end{lemma}
\begin{proof}
	Let $h$ be the result of $H$ and $g$ be the result of $G$.
	Suppose that $q_j$ performs an \id{update} after $H$ and before $G$.
	Since every \id{update} appends triples with a new timestamp to $\mathsf{H}[j]$,
	$\# g_j > \# h_j$ will hold on~\lineref{line:xheqhh} in $B$,
	and the \id{Block-Update} $B$ must return \peace{}.
\end{proof}
\begin{lemma}
	\label{lem:xblocklin}
	Let $B$ be a \id{Block-Update} operation by $q_i$ that does not return \peace{}
	and let $X$ be the \id{update} on~\lineref{line:xupdate} in $B$.
	Then, all \id{Update}s in $B$ are linearized at $X$, consecutively,
	in order of the components they update.
\end{lemma}
\begin{proof}
	Let $H$ and $G$ be the \id{scan} operations on~\lineref{line:xh} and~\lineref{line:xhh} in $B$, 
	respectively, and let $h$ be the result of $H$.
	Consider the timestamp $t = \textsc{New-Timestamp}(h)$ associated with $B$.  
	Suppose some \id{update} to $\mathsf{H}$ before $H$ appends a triple with 
	a timestamp $s$.
	Then, $h$ contains this triple with timestamp $s$ and,
	by~\corollaryref{cor:tsdom}, $t \succ s$.
	
	Consider any \id{update} $X'$ that has appended an update triple with timestamp $s \succ t$.
	If $X'$ occurs before $X$, then $X'$ occurs between $H$ and $X$.
	Let $B'$ be the \id{Block-Update} that contains $X'$,
	let $H'$ be the \id{scan} of $\mathsf{H}$ in $B'$ 
	on~\lineref{line:xh} from which $s$ is generated and let $h'$ be the result of $H'$.
	$B'$ is concurrent with $B$ and thus, not performed by $q_i$.
	If $s_i \geq t_i$, then, since $s_i = \#h'_i$, we have $\#h'_i \geq t_i$,
	implying that $H'$ occurs after $X$.
	But this is impossible, since $H'$ occurs before $X'$.
	Therefore $s_i < t_i$.
	
	Since $s \succeq t$, there exists $j < i$ such that $s_j > t_j$.
	This is only possible if process $q_j$ performed an \id{update}
	after $H$ and before $H'$, or if $B'$ is performed by $q_j$.
	In the first case, since $H'$ occurs before $X'$, which occurs before $X$,
	which occurs before $G$, this contradicts~\lemmaref{lem:xbots}.
	In the second case, $X'$ is an \id{update} by $q_j$ between $H$ and $X$.
	Since $X$ occurs before $G$, this also contradicts~\lemmaref{lem:xbots}.
	
	Thus, all \id{update}s with timestamp $s \succ t$ occur after $X$.
	All \id{Update}s that are part of $B$ have the same timestamp $t$.
	Therefore, all \id{Update}s by $B$ are linearized at $X$.
	By~\lemmaref{lem:tsunique}, timestamps are unique.
	\id{Update}s linearized at the same point are ordered 
	first by their timestamps and then by the components they update.
	Hence, all \id{Update}s that are part of $B$
	will be ordered consecutively, sorted in order of their components.
\end{proof}
Next, let us consider \id{Block-Update} operations that return \peace{}.
\begin{lemma}
	\label{lem:splitblock}
	Let $B$ be a \id{Block-Update} operation that returns \peace{},
	let $H$ be the \id{scan} of $\mathsf{H}$ on~\lineref{line:xh} in $B$ with result $h$ and
	let $X$ be the \id{update} on~\lineref{line:xupdate}.
	Then all \id{Update}s in $B$ are linearized after $H$ and no later than $X$.
\end{lemma}
\begin{proof}
	Let $U$ be an \id{Update} to component $j$ with associated timestamp $t$ that is part of $B$.
	$U$ is linearized at the first point that $\mathsf{H}$ contains
	an update triple with component $j$ and timestamp $t' \succeq t$.
	Note that $t$ is generated from $h$ on~\lineref{line:xts} in $B$.
	By~\corollaryref{cor:tsdom}, all of the timestamps contained in $h$ are lexicographically 
	smaller than $t$.
	Thus, $U$ is linearized after $H$.
	Since $X$ appends an update triple with component $j$ and timestamp $t$,
	$U$ is linearized no later than $X$ by~\observationref{obs:xupdlin}.
\end{proof}
Thus, every \id{Block-Update} is linearized within its execution interval.
\begin{lemma}
	\label{lem:minblock}
	Let $B$ be a \id{Block-Update} by $q_i$ whose execution interval does not contain
	any \id{update}s by a process $q_j$ to $\mathsf{H}$ on~\lineref{line:xupdate} with $j < i$.
	Then, $B$ does not return \peace{}.
\end{lemma}
\begin{proof}
	Suppose $B$ is returns \peace{}.
	Let $h$ and $g$ be the results of the scans of $\mathsf{H}$ on~\lineref{line:xh} 
	and~\lineref{line:xhh}, respectively, in $B$.
	Then, for some $j < i$, $\# g_j > \# h_j$.
	This implies that $q_j$ has performed an \id{update} on $\mathsf{H}$ between $H$ and $G$.
\end{proof}
Next, we show that our choice of linearization points for $\id{Scan}$s and $\id{Update}$s
produces a valid linearization.
\begin{lemma}
	\label{lem:xscanok}
	Let $H$ be a \id{scan} that returns $h$. 
	Suppose $\id{Get-view}(h) = (v_1,\dots,v_m)$.
	Then, for each $1 \leq j \leq m$, $v_j$ is the value of the last \id{Update} to component
	$j$ of $\mathcal{M}$ linearized before $H$, or $\bot$ if no such \id{Update} exists.
\end{lemma}
\begin{proof}
	Suppose that $h$ contains an update triple involving component $j$.
	This triple was appended to $\mathsf{H}$ by some update $X$ that is part of 
	a $\id{Block-Update}$ $B$.
	By~\lemmaref{lem:xblocklin} and~\lemmaref{lem:splitblock}, all $\id{Update}s$ in $B$
	are linearized at or before $X$.
	Hence, if no \id{Update} to component $j$ is linearized before $H$, then $v_j = \bot$.
	
	Now, consider the last \id{Update} $U$ to component $j$ linearized before $H$.
	Let $t$ be its associated timestamp.
	Let $t'$ be the largest timestamp of any update triple  with component $j$ in $h$.
	By~\observationref{obs:xtscontain}, $t' \succeq t$.
	By~\lemmaref{lem:tsunique}, there is exactly one update triple in $h$ with component $j$ and timestamp $t'$. 
	By definition of $\id{Get-view}(h)$, $v_j$ is the value of this update triple.
	Let $X'$ be the \id{update} to $\mathsf{H}$ that appended $(j,v_j,t')$ during
	a \id{Block-Update} operation $B'$ and let $U'$ be the \id{Update} to component $j$ in $B'$.
	Since $(j,v_j,t')$ is contained in $h$, $X'$ occurs before $H$.
	By definition of $t'$, 
	$U'$ is linearized at $X'$.
	
	Since $t' \succeq t$, by~\observationref{obs:xupdlin},
	$U$ is linearized at no later than $X'$. 
	By definition $U$ is the last \id{Update} to component $j$ linearized before $H$. 
	Since $U'$ is linearized at $X'$, $U$ is linearized at $X'$ and $t \succeq t$.
	Therefore, $t = t'$, which by~\lemmaref{lem:tsunique} implies that $U = U'$.
\end{proof}
\begin{corollary}[\id{Scan}s]
	\label{cor:xscanok}
	Consider any $\id{Scan}$ that returns $(v_1,\dots,v_m)$.
	Then, for each $1 \leq j \leq m$, $v_j$ is the value of the last \id{Update} to component
	$j$ of $\mathcal{M}$ linearized before the \id{Scan} operation, 
	or $\bot$ if no such \id{Update} exists.
\end{corollary}
We now consider the linearization of $\id{Block-Update}$s.
Suppose $B$ is a \id{Block-Update} that does not return \peace{}.
Throughout the rest of this section, we use $H$, $X$, $H'$, $\ell$, and $L$ as follows.
Let $H$ be the \id{scan} of $\mathsf{H}$ in $B$ on~\lineref{line:xh},
let $X$ be the \id{update} in $B$ on~\lineref{line:xupdate},
let $H'$ be the \id{scan} in $B$ on~\lineref{line:xhh},
let $\ell$ be the value of $\mathit{last}$ when $B$ returns on~\lineref{line:xbview},
and let $L$ be the last \id{scan} of $\mathsf{H}$ that returns $\ell$. 
\begin{lemma}
	\label{lem:xlastscan1}
	Consider any \id{Block-Update} operation $B$ that does not return \peace{}. 
	Then $L$ occurs no earlier than $H$ and before $X$.
\end{lemma}
\begin{proof}
	Suppose $B$ is performed by process $q_i$.
	Let $h$ be the result of $H$ and let $r[j]$ be the value \id{read} from $\mathsf{L}_{j,i}[\#h_i]$ 
	for $j \in \{1,\ldots,f\} - \{i\}$ on~\lineref{line:xreadLj} during $B$.
	By~\lineref{line:xwriteL} and~\lineref{line:xlow}, a process $q_j \neq q_i$ only \id{writes} 
	to $\mathsf{L}_{j,i}[\#h_i]$ when it takes a \id{scan} $G$ of $\mathsf{H}$ with 
	result $g$ such that $\#g_i = \#h_i$.
	$X$ appends triples with a new timestamp to $\mathsf{H}[i]$, so any \id{scan} $G'$ of
	$\mathsf{H}$ performed after $X$ returns a result, $g'$, such that $\# g'_i > \# h_i$.
	Thus, 
	if $r[j] \neq \bot$, then $r[j]$ is the result of a \id{scan} of
	$\mathsf{H}$ performed before $X$.
	
	By~\lineref{line:xlastinit},~\lineref{line:xlastmax}, and~\lineref{line:xlastmax2}, 
	$\ell \in \{h, r[1], \ldots, r[{i-1}], r[{i+1}], \ldots, r[f]\}$, $\# \ell_i = \# h_i$, 
	and $h$ is a prefix of $\ell$. 
	Hence, any \id{scan} that returns $\ell$, in particular $L$, occurs before $X$.
	If $h$ is a proper prefix of $\ell$, then~\observationref{obs:prefix-obs} 
	implies that $L$ occurs no earlier than $H$.
	Otherwise, if $h = \ell$, $L$ occurs no earlier than $H$ as $L$ is the last \id{scan} that returns $\ell$.
\end{proof}
By~\lemmaref{lem:xlastscan1}, $L$ occurs no earlier than $H$ and before $X$,
and thus the interval starting immediately after $L$ and ending with $X$ 
is contained within $B$'s execution interval.
We call this interval the \emph{window} of $B$. 

\begin{lemma}
	\label{lem:xlastscan2}
	Consider any \id{Block-Update} operation $B$ that does not return \peace{}. 
	Then, no \id{Scan} operation is linearized during the window of $B$. 
\end{lemma}
\begin{proof}
	For a contradiction, suppose that a \id{Scan} operation $S$ is linearized 
	in the window of $B$.
	Let $G$ be the last \id{scan} in $S$, performed on~\lineref{line:xlastscan}, 
	and let $g$ be the result of $G$.
	By definition, $G$ is the linearization point of $S$, which, by assumption,
	occurs during the window of $B$.
	It follows that $S$ is not performed by $q_i$,
	which performs $X$ as its first step after $H$.
	Let $q_j \neq q_i$ be the process that performs $S$.
	
	$G$ occurs after $L$, which occurs no earlier than $H$.
	Thus, by~\observationref{obs:prefix-obs} we have $\#g_i \geq \#h_i$.
	Since $G$ occurs before $X$, by~\observationref{obs:xbeforex} we have $\#g_i \leq \#h_i$, 
	so $\#g_i = \#h_i$.
	
	By~\lineref{line:xwriteL} and~\lineref{line:xscancond}, 
	$q_j$ wrote $g$ to $\mathsf{L}_{j,i}[\#h_i]$ prior to $G$.
	$G$ occurs before $X$ and $q_i$ \id{reads} $r_j$ from $\mathsf{L}_{j,i}[\#h_i]$ after $X$.
	Thus, by~\lemmaref{lem:xhelp}, $g$ is a prefix of $\ell$.  
	Since $L$ occurs before $G$, $\ell$ is a prefix of $g$.
	Therefore, $g = \ell$. 
	However, $G$ occurs after $L$, contradicting the definition of $L$.
\end{proof}
\begin{lemma}
	\label{lem:xlastscan3}
	The windows of \id{Block-Update}s that do not return \peace{} are pairwise disjoint. 
\end{lemma}
\begin{proof}
	Assume to the contrary that the windows of two \id{Block-Update} operations $B$ and $B'$
	that do not return \peace{} do intersect.
	Let $H(B')$, $X(B')$, and $H'(B')$ be defined in a similar fashion for $B'$
	as $H$, $X$, $H'$ for $B$.
	In particular, let $H(B')$ be the \id{scan} of $\mathsf{H}$ in $B'$ on~\lineref{line:xh},
	let $X(B')$ be the \id{update} in $B'$ on~\lineref{line:xupdate}, and
	let $H'(B')$ be the \id{scan} in $B'$ on~\lineref{line:xhh}.
	
	Suppose $B$ is performed by process $q_i$, and $B'$ is performed by $q_j \neq q_i$.
	Without loss of generality, suppose that $X(B')$ occurs before $X$.
	Since the windows of $B$ and $B'$ intersect, 
	$X(B')$ occurs after $L$.
	By~\lemmaref{lem:xlastscan1}, $L$ occurs no earlier than $H$. 
	\lemmaref{lem:xbots} applied to $B$ implies that $j > i$,
	as $X(B')$ is an \id{update} by $q_j$ that occurs between $H$ and $H'$.
	Since $H(B')$ occurs before $X(B')$, which occurs before $X$,~\lemmaref{lem:xbots}
	applied to $B'$ implies that $H'(B')$ occurs before $X$.
	
	Let $G$ be the \id{scan} on~\lineref{line:xhelp} in $B'$ with result $g$.
	$G$ occurs before $H(B')$, which occurs before $X$.
	By~\observationref{obs:xbeforex}, we get $\# g_i \leq \# h_i$.
	On the other hand, $G$ occurs after $X(B')$, which occurs after $L$,
	and $L$ occurs no earlier than $H$.
	Thus, by~\observationref{obs:prefix-obs}, $\# g_i \geq \# h_i$.
	Hence, $\#g_i = \#h_i$.
	
	In $B'$, process $q_j$ \id{write}s $g$ to $\mathsf{L}_{j,i}[\#g_i] = \mathsf{L}_{j,i}[\#h_i]$
	on~\lineref{line:xlow} before $H'(B')$.
	$H'(B')$ occurs before $X$ and $q_i$ reads $r[j] = \mathsf{L}_{j,i}[\#h_i]$ 
	on~\lineref{line:xreadLj} after $X$.
	Thus, by~\lemmaref{lem:xhelp}, $g$ is a prefix of $\ell$.
	Since $L$ occurs before $X(B')$, which occurs before $G$, 
	$\ell$ is also a prefix of $g$.
	Therefore, $\ell = g$.
	However, $G$ occurs after $L$, contradicting the definition of $L$.
\end{proof}
Combining the last few lemmas, we prove that \id{Block-Update}s return correct values.
\begin{lemma}[\id{Block-Update}s]
	\label{lem:xblock}
	Consider any \id{Block-Update} operation $B$ by $q_i$ that does not return \peace{}.
	Let $Z$ be the first linearization point of $B$'s \id{Update}s
	and let $Z'$ be the linearization point of the last \id{Update} prior to $Z$ from 
	a \id{Block-Update} $B'$ that does not return \peace{},
	or the beginning of the execution if all \id{Block-Update}s prior to $Z$ return \peace{}.
	$B$ returns the values of all components of $\mathcal{M}$ at $L$, which 
	occurs between $Z'$ and $Z$. 
	Only \id{Update}s from \id{Block-Update}s that return \peace{} by processes $q_j \neq q_i$ 
	are linearized between $L$ and $Z$.
\end{lemma} 
\begin{proof}
	On~\lineref{line:xbview}, $B$ returns \id{Get-view(\ell)}, which
	by~\lemmaref{lem:xscanok} contains the values of all components of $\mathcal{M}$ at $L$.
	
	By~\lemmaref{lem:xblocklin}, $Z = X$.
	By~\lemmaref{lem:xlastscan1}, $L$ occurs no earlier than $H$ and before $Z$. 
	Recall that the window of $B$ starts immediately after $L$ 
	and ends with $Z$.
	Since $q_i$ performs $Z$ as its first step after $H$ in $B$ and $B$ is linearized at $X$,
	no operation by $q_i$ can be linearized between $L$ and $Z$.
	By~\lemmaref{lem:xlastscan2}, no \id{Scan} operation is 
	linearized between $L$ and $Z$.
	
	If $Z'$ is the linearization point of the last \id{Update} prior to $Z$ from 
	a \id{Block-Update} $B'$ that does not return \peace{},
	then, by~\lemmaref{lem:xblocklin}, $Z'$ is the \id{update} on~\lineref{line:xupdate} in $B'$.
	Hence, by definition, $Z'$ is the end of the window of $B'$.
	By~\lemmaref{lem:xlastscan3}, windows of $B$ and $B'$ are disjoint.
	Since $Z$ occurs after $Z'$, it follows that $Z'$ occurs before $L$.
	If $Z'$ is the beginning of the execution, then $Z'$ also occurs before $L$.
	
	By definition of $Z'$, no \id{Update} from a \id{Block-Update} $B'$ that does 
	not return \peace{} is linearized after $Z'$ and before $Z$, and hence, between $L$ and $Z$.
\end{proof}

\begin{theorem}
	\label{thm:fulnb}
	There is a non-blocking implementation
	of a augmented $m$-component multi-writer snapshot object. 
	A \id{Block-Update} by $q_i$ returns \peace{}
	only if its execution interval contains an \id{update} by a process $q_j$ with $j < i$ 
	(performed as a part of a \id{Block-Update} by $q_j$).
\end{theorem}
\begin{proof}
	From the code, \id{Block-Update} operations are wait-free.
	If a process takes steps but does not return from an invocation of \id{Scan},
	then the test on~\lineref{line:xscancond} must repeatedly fail.
	This is only possible if a new triple is appended to $\mathsf{H}$
	by an \id{update} on~\lineref{line:xupdate}. 
	Since each \id{Block-Update} operation performs only one \id{update} to $\mathsf{H}$,
	other processes must be completing infinitely many invocations of \id{Block-Update}. 
	
	Second part of the theorem follows from~\lemmaref{lem:minblock}.
\end{proof}

%
\section{The Simulation}
\label{sec:simulation}
	
In this section, we prove the main result of our paper:

\begin{theorem} [Simulation]
	\label{thm:simulation}
	Let $T$ be a colorless task, let $f \leq n$, and let $\Pi$ be a protocol solving $T$
	among $n$ processes using an $m$-component multi-writer snapshot.
	\begin{itemize}
		\item If $\Pi$ is obstruction-free and $L$ is a lower bound on the step complexity of solving $T$ in a wait-free manner among $f$ processes 
		using a single-writer snapshot, then 
		$m \geq \min \left\{ \lfloor \frac{n}{f} \rfloor + 1, (\log_2 \frac{L}{f})^{\frac{1}{2}} \right \}$. 
		\item If $\Pi$ is $x$-obstruction-free, for some $1 \leq x < f$, and $T$ cannot be solved in a wait-free manner among $f$ processes using a single-writer snapshot, then $m \geq \lfloor \frac{n-x}{f-x}\rfloor + 1$.
	\end{itemize}
\end{theorem}

The bound is derived by considering a protocol $\Pi$ for solving a colorless task $T$ among $n$ processes  where $m$ is too small. We show how $f$ processes can simulate this protocol in a wait-manner.
Furthermore, if $L$ is a lower bound on the step complexity of solving $T$ in a wait-free manner, then we show that the step complexity of the simulation is less than $L$. 

In  our simulations, there are $0 \leq d < f$ \emph{direct} simulators and $f-d$ \emph{covering} simulators. We ensure covering simulators have smaller identifiers than direct simulators. Each simulator $q_i$ is
responsible simulating a set of processes $P_i$. If $q_i$ is a direct simulator, then $|P_i| = 1$. Otherwise, $|P_i| = m$. Crucially, each simulated process is simulated by at most one simulator, i.e.~for all $i \neq j$, $P_i$ and $P_j$ are disjoint.

To prove the first case, we consider the simulation with $d = 0$ direct simulators. Here, we show that, if $\Pi$ uses $m < \min \{ \lfloor \frac{n}{f} \rfloor +1, (\log_2 \frac{L}{f})^{\frac{1}{2}} \}$ components, then we can bound the step complexity of the simulation from above by $2^{fm^2} < L$.
For the second case, we consider the simulation with $d = x$ direct simulators and $m < \lfloor \frac{n-x}{f-x} \rfloor + 1$. 
In all cases, the bound on $m$ implies that $(f-x)m + x \leq n$, i.e.~there are enough processes for the simulators to simulate.

As discussed in the preliminaries, without loss of generality, we assume that: 

\begin{assumption}
	\label{ass:altscanupdate}
	In the protocol $\Pi$, each process alternately performs $\id{scan}$ and $\id{update}$ operations on the snapshot object, $M$, until it performs a $\id{scan}$ that allows it to output a value. 
\end{assumption}

\subsection{Simulation Algorithm}

In this section, we describe the simulation algorithms of the direct and covering simulators. Both direct and covering simulators use a non-blocking implementation of
a shared $m$-component augmented 
snapshot object, $\real{M}$,
for simulating the steps (i.e.,~\id{update}s and \id{scan}s on $M$) of the processes they are simulating. We use $\real{M}.\id{Block-Update}$ and $\real{M}.\id{Scan}$ to denote operations on $\real{M}$ and refer to $\real{M}.\id{Update}$s that are part of $\real{M}.\id{Block-Update}$ operations. We use $M.\id{update}$ and $M.\id{scan}$ to denote operations on $M$. Finally, we will say that simulators \emph{apply} operations on $\real{M}$ while the simulated processes \emph{perform} operations on $M$. All other variables in the algorithms are local.

A direct simulator
directly simulates its single process in a step-by-step manner.
A covering simulator attempts to simulate its set of processes so that they 
all cover different components of $\simulated{M}$.
The manner in which it does so resembles a covering argument: 
it tries to simulate its processes so that they perform 
block updates and cover successively more components.
Analogously, this
involves inserting \emph{hidden} steps by some simulated processes, which
are \emph{locally} simulated, i.e.~without performing any operations on $\real{M}$.

We will guarantee that, for each \emph{real execution} of the simulators 
(i.e. an execution by the real processes in the real system),
there exists a corresponding \emph{simulated execution} of the protocol $\Pi$ (by the simulated processes in the simulated system).
However, because of the locally simulated steps, 
the exact correspondence between these executions is too complex
to be described here without proper formalism.

\textbf{Direct simulator's algorithm.} A direct simulator $q_i$ \emph{directly} simulates its single process $p_{i,1} \in P_i$ as follows. Initially, $q_i$ sets the input of $p_{i,1}$ to its input, $x_i$. To simulate an $M.\id{scan}$ by $p_{i,1}$, $q_i$ applies an $\real{M}.\id{Scan}$. To simulate an $M.\id{update}(j,v)$ by $p_{i,1}$, $q_i$ applies a one component $\real{M}.\id{Block-Update}([j],[v])$, ignoring the value returned. At any point, if $p_{i,1}$ outputs some value $y$ and terminates, then $q_i$ outputs $y$ and terminates. The pseudocode appears in~\algorithmref{alg:directsimulator}. 

\begin{algorithm} 
	\begin{algorithmic} [1]
		\State initialize $p_{i,1}$'s input to $x_i$
		\Loop
		\State simulate $p_{i,1}$'s next step (which is an $M.\id{scan}$) using $\real{M}.\id{Scan}$ and update its state
		\If {$p_{i,1}$ is poised to perform $M.\id{update}(j,v)$}
		\State \,\, simulate $p_{i,1}$'s next step using $\real{M}.\id{Block-Update}([j],[v])$ and update its state
		\Else \, $\triangleright$ {$p_{i,1}$ has output some value $y$}
		\State \,\, \textbf{output} $y$ and \textbf{terminate}
		\EndIf
		\EndLoop 
	\end{algorithmic}
	\caption{Pseudocode for a direct simulator $q_i$ on input $x_i$.}
	\label{alg:directsimulator}
\end{algorithm}


\textbf{Covering simulator's algorithm.} A covering simulator $q_i$ applies an $\real{M}.\id{Block-Update}$ operation, $\real{B}$, to attempt to simulate a block update by a subset of the processes in $P_i$. If $\real{B}$ returns a view $\real{V} \neq \peace$, then, by the specification of the augmented snapshot, $\real{M}$, $q_i$ knows that $\real{B}$ was \emph{atomic}, i.e.~the individual $\real{M}.\id{Update}$s in $\real{B}$ can be linearized consecutively. Moreover, $q_i$ knows that $\real{V}$ is a view of $\real{M}$ at some earlier point $\bm{t}$ in the real execution such that there are no $\real{M}.\id{Scan}$s or $\real{M}.\id{Update}$s that are part of atomic $\real{M}.\id{Block-Update}$s linearized between $\bm{t}$ and the linearization point of the first $\real{M}.\id{Update}$ that is part of $\real{B}$. 

Given this knowledge, at some later point $\bm{t}'$ in the real execution, $q_i$ may choose to \emph{revise the past} as follows. First, $q_i$ picks a process $p \in P_i$ such that it has not simulated any steps of $p$ between $\bm{t}$ and $\bm{t}'$, i.e.~the state of $p$ that it currently stores at $\bm{t}'$ is also the state of $p$ that it stored at $\bm{t}$. Then it \emph{locally} simulates a solo execution $\xi$ of $p$ using its current state of $p$, assuming that the contents of $M$ are the same as $\real{V}$. We will guarantee that, at the point $t$ corresponding to $\bm{t}$ in the simulated execution, the contents of $M$ are indeed $\real{V}$ and that the state of $p$ at $t$ is the same as at $\bm{t}$. Hence, this has the effect of inserting $\xi$ immediately after $t$ in the simulated execution. Finally, to ensure that the resulting simulated execution is valid, $q_i$ ensures that $\xi$ only contains $M.\id{update}$s to components updated by $\real{B}$ and $M.\id{scan}$s. Hence, the steps in $\xi$ are hidden by the block update corresponding to $\real{B}$ in the simulated execution and $p$ \emph{could have} taken those steps immediately after $t$. In this case, we say that, \emph{at $\bm{t}'$, $q_i$ revised the past of $p$ using $\real{V}$}.
On the other hand, if $\real{B}$ returns $\peace$, then $q_i$ knows that the \id{Update} operations comprising
$\real{B}$ have occurred, but not necessarily consecutively. So, $q_i$ cannot use $\real{B}$ to hide steps by any of its simulated processes. 

To describe the algorithm of a covering simulator $q_i$, we fix a labelling $p_{i,1},\dots,p_{i,m}$ of the 
processes $P_i$ that $q_i$ simulates. Initially, $q_i$ sets the input of each process in $P_i$ to its input $x_i$. 
The goal of $q_i$ is to \emph{construct} a block update by $P_i$ to all $m$ components of $M$, i.e.~simulate the processes in $P_i$ so that, eventually, $P_i$ covers all components of $M$. To do so, $q_i$ recursively constructs and simulates block updates by $p_{i,1},\dots,p_{i,r}$ to $r$ components of $M$, for increasing $1 \leq r < m$. At any point in $q_i$'s construction, if a process in $P_i$ outputs some value $y$ and terminates, then $q_i$ outputs $y$ and terminates, without further simulating the rest of the processes in $P_i$.

As a base case, to construct a block update to a single component, $q_i$ simulates the next step of $p_{i,1}$, which we will ensure is an $M.\id{scan}$, using $\real{M}.\id{Scan}$. If $p_{i,1}$ is poised to perform $M.\id{update}(j_1,v_1)$ after this, then $q_i$ constructs the block update $M.\id{update}(j_1,v_1)$. Otherwise, $p_{i,1}$ has output some value $y$, so $q_i$ outputs $y$ and terminates.

To construct a block update to $r > 1$ components, $q$ constructs a sequence of block updates $\beta_1',\beta_2',\dots$, each to $r-1$ components, and simulates them using $\real{M}.\id{Block-Update}$ operations. It continues until 
(one of its simulated processes terminates or)
it constructs a block update $\beta_t'$ to $r-1$ components
that updates the same set of components as some block update $\beta_{s}'$ for $s < t$, which was simulated by an \emph{atomic} $\real{M}.\id{Block-Update}$, i.e.~it returns a view $\real{V} \neq \peace$.
Let $j_1,\dots,j_{r-1}$ be the components that $\beta_t'$ updates and let $v_1,\dots,v_{r-1}$ be the values to which it updates these components. After constructing $\beta_t'$, $q_i$ revises the past of $p_{i,r}$ using $\real{V}$, i.e.~it continues its simulation of $p_{i,r}$ by \emph{locally} simulating a solo execution of $p_{i,r}$, assuming that the contents of $M$ are $\real{V}$ at the beginning of this execution. It does so until $p_{i,r}$ 
is about to perform an $M.\id{update}$ to a component $j_{r} \notin \{j_1,\dots,j_{r-1}\}$ with some value $v_r$ (or $p_{i,r}$ terminates). 
If $p_{i,1},\dots,p_{i,r}$ do not terminate, then $q_i$ has constructed the block update $\beta_t \cdot M.\id{update}(j_r,v_r)$. The pseudocode appears in \algorithmref{alg:coveringconstruction}.

\begin{algorithm} 
	\begin{algorithmic}[1]
		\Function{Construct}{$r$} 
		\If {$r = 1$} \, $\triangleright$ base case
		\State simulate $p_{i,1}$'s next step (which is an $M.\id{scan}$) using $\real{M}.\id{Scan}$ and update its state
		\If {$p_{i,1}$ is poised to perform $M.\id{update}(j,v)$}
		\State \,\,\, \Return $([j],[v])$
		\Else \, $\triangleright$ $p_{i,1}$ has output some value $y$
		\State \,\,\, \textbf{output} $y$ and \textbf{terminate}
		\EndIf
		\Else \, $\triangleright$ $r > 1$
		\State $A \gets \emptyset$  \Comment{$A$ contains pairs $(J,\real{V})$, where $J$ is a set of $r-1$ components and $\real{V}$ is a view}
		\Loop
		\State $([j_1,\dots,j_{r-1}],[v_1,\dots,v_{r-1}]) \gets \Call{Construct}{r-1}$ 
		\If {there exists $(J',\real{V}') \in A$ such that $J' = \{j_1,\dots,j_{r-1}\}$} 
		\State \textbf{locally} simulate $p_{i,r}$ assuming contents of $M$ are $\real{V}'$ (updating its state) 
		\Statex \hspace{2.2cm} \textbf{until} $p_{i,r}$ is poised to perform $M.\id{update}$ to a component not in $\{j_1,\dots,j_{r-1}\}$ 
		\Statex \hspace{2.7cm}  \textbf{or}  $p_{i,r}$ has output a value
		\If {$p_{i,r}$ poised to perform $M.\id{update}(j_{r},v_{r})$}
		\State \,\, \Return $([j_1,\dots,j_{r-1},j_{r}],[v_1,\dots,v_{r-1},v_{r}])$
		\Else \, $\triangleright$ $p_{i,r}$ has output some value $y$ 
		\State \,\, \textbf{output} $y$ and \textbf{terminate}
		\EndIf
		\Else 
		\State simulate $p_{i,1},\dots,p_{i,r-1}$'s next steps using 
		\Statex \hspace{2.7cm} $\real{M}.\id{Block-Update}([j_1,\dots,j_{r-1}],[v_1,\dots,v_{r-1}])$ and update their states
		\If {this $\id{Block-Update}$ returns a view $\real{V} \neq \peace{}$} \, $\triangleright$ {i.e., it is atomic}
		\State $A \gets A \cup \{(\{j_1,\dots,j_{r-1}\},\real{V})\}$
		\EndIf
		\EndIf
		\EndLoop
		\EndIf
		\EndFunction
	\end{algorithmic}
	\caption{Pseudocode for covering simulator $q_i$ to construct a block updates to $r$ components, where $1 \leq r \leq m$. Assumes that $p_{i,1}$ is poised to perform $M.\id{scan}$. Returns a block update by $p_{i,1},\dots,p_{i,r}$, represented as a pair $([j_1,\dots,j_r],[v_1,\dots,v_r])$, where $p_{i,g}$ is poised to perform $M.\id{update}(j_g,v_g)$, for $1 \leq g \leq r$.}
	\label{alg:coveringconstruction}
\end{algorithm}

If $q_i$ constructs a block update $\beta$ to $m$ components, 
then  $q_i$ locally simulates $\beta$ followed by
the terminating solo execution, $\xi$, of $p_{i,1}$. 
Then process $q_i$  terminates and outputs the value that $p_{i,1}$
outputs in $\xi$.
Notice that $\beta\xi$ is 
applicable at any point after $\beta$ has been constructed, 
since the block update
completely overwrites 
$M$. 
In fact, these steps will occur at the end of the final simulated execution.
The pseudocode appears in~\algorithmref{alg:coveringsimulator}.

\begin{algorithm} 
	\begin{algorithmic}[1]
		\State initialize $p_{i,1},\dots,p_{i,m}$'s inputs to $x_i$
		\State $\beta \gets \Call{Construct}{m}$ 
		\State store the states of $p_{i,1},\dots,p_{i,m}$
		\State \textbf{locally} simulate $p_{i,1}$'s terminating solo execution, $\xi$, after $\beta$ 
		\State restore the states of $p_{i,1},\dots,p_{i,m}$
		\State \textbf{output} $p_{i,1}$'s output in $\xi$ and \textbf{terminate}
	\end{algorithmic}
	\caption{Pseudocode for a covering simulator $q_i$ on input $x_i$.}
	\label{alg:coveringsimulator}
\end{algorithm}

\subsection{Properties of Covering Simulators}

In this section, we prove properties of the covering simulator's algorithm. To do so, we first consider the procedure, $\textsc{Construct}$, which is used by the covering simulators to construct block updates. 

\begin{proposition}
	\label{prop:construct-procedure}
	Let $1 \leq r \leq m$ and let $\textsc{Q}$ be a call to $\textsc{Construct}(r)$ by a covering simulator $q_i$. If $p_{i,1}$ is poised to perform $M.\id{scan}$ when it is in the state stored by $q_i$ immediately before $\textsc{Q}$, then the following properties hold:
	\begin{enumerate}
		\item \label{prop:construct-procedure-altscanblock} During $\textsc{Q}$, $q_i$ alternately applies $\real{M}.\id{Scan}$s and $\real{M}.\id{Block-Update}$s to at most $r-1$ components, starting with at least one $\real{M}.\id{Scan}$. Each $\real{M}.\id{Scan}$ simulates an $M.\id{scan}$ by $p_{i,1}$ and each $\real{M}.\id{Block-Update}$ to $s \leq r-1$ components simulates $M.\id{update}$s by $p_{i,1},\dots,p_{i,s}$. In particular, $q_i$ does not apply any $\real{M}.\id{Block-Update}$s in $\textsc{Construct}(1)$. 
		\item \label{prop:construct-procedure-terminate} If $q_i$ outputs some value $y$ (and, hence, terminates) during $\textsc{Q}$, then the last operation $q_i$ applied was $\real{M}.\id{Scan}$ and one of $q_i$'s simulated process $p_{i,g} \in P_i$, for some $1 \leq g \leq r$, has output $y$. 
		\item \label{prop:construct-procedure-return} If $q_i$ returns $[(j_1,\dots,j_r),(v_1,\dots,v_r)]$ from $\textsc{Q}$, then the last operation $q_i$ applied was $\real{M}.\id{Scan}$. Moreover, in $\textsc{Q}$, $q_i$ revises the past of $p_{i,r}$ immediately after this $\real{M}.\id{Scan}$ so that it is poised to perform $M.\id{update}(j_r,v_r)$. For $r < g \leq m$, the state of $p_{i,g}$ does not change as a result of the call.
		\item \label{prop:construct-procedure-blocks} Suppose $r > 1$ and $q_i$ does not terminate in any call to $\textsc{Construct}(r-1)$ during $\textsc{Q}$. Let $\real{\delta}$ be the last $\real{M}.\id{Scan}$ in $\textsc{Q}$. Then, during $\textsc{Q}$, $q_i$ applied a sequence of atomic $\real{M}.\id{Block-Update}$s $\real{B}'_{r-1},\dots,\real{B}_{1}'$ such that, for each $1 \leq g \leq r-1$, $\real{B}_{g}'$ updated $g$ components, $q_i$ revised the past of $p_{i,g+1}$ using the view returned by $\real{B}_g'$ immediately after $\real{\delta}$, and, from $\real{B}_{g}'$ until $q_i$ outputs some value or returns from $\textsc{Q}$, $q_i$ does not apply any $\real{M}.\id{Block-Update}$s to $g+1$ or more components.
	\end{enumerate}
\end{proposition}
\begin{proof}
	
	By induction on $r$. The base case is when $r = 1$. Observe that, in $\textsc{Construct}(1)$, $q_i$ applies a single $\real{M}.\id{Scan}$. After simulating $p_{i,1}$'s next step using  this $\real{M}.\id{Scan}$, if $p_{i,1}$ is poised to perform $M.\id{update}(j_1,v_1)$, then $q_i$ returns $([j_1],[v_1])$ from the call. Otherwise, by~\assumptionref{ass:altscanupdate}, the $M.\id{scan}$ allows $p_{i,1}$ to output some value $y$, so $q_i$ outputs $y$ and terminates. It follows that the first three parts of the claim holds. The fourth part of the claim is vacuously true. Now let $r > 1$ and suppose the claim holds for $r-1$. 
	
	Since $p_{i,1}$ is poised to perform $M.\id{scan}$ when $q_i$ calls $\textsc{Construct}(r)$. Hence, by the code, when $q_i$ recursively calls $\textsc{Construct}(r-1)$ for the first time in $\textsc{Construct}(r)$, $p_{i,1}$ is still poised to perform $M.\id{scan}$. It follows that we may apply the induction hypothesis to conclude that, during the first call to $\textsc{Construct}(r-1)$, $q_i$ alternately applies $\real{M}.\id{Scan}$ and $\real{M}.\id{Block-Update}$, starting with at least one $\real{M}.\id{Scan}$ and ending with an $\real{M}.\id{Scan}$. Moreover, if $q_i$ returns from the first call to $\textsc{Construct}(r-1)$, then $p_{i,1},\dots,p_{i,r-1}$ are poised to perform $M.\id{update}$s. By the code, each subsequent call to $\textsc{Construct}(r-1)$ is immediately preceded by an $\real{M}.\id{Block-Update}$, which $q_i$ applied to simulate the $M.\id{update}$s that $p_{i,1},\dots,p_{i,r-1}$ were poised to perform as a result of the previous call to $\textsc{Construct}(r-1)$. By~\assumptionref{ass:altscanupdate}, this implies that, immediately before each subsequent call to $\textsc{Construct}(r-1)$, $p_{i,1}$ is poised to perform $M.\id{scan}$ and the induction hypothesis is applicable to the call. It follows that, during $\textsc{Construct}(r)$, $q_i$ alternately applies $\real{M}.\id{Scan}$ and $\real{M}.\id{Block-Update}$, starting with at least one $\real{M}.\id{Scan}$. Hence, the first part of the claim holds.

	
	If $q_i$ outputs some value $y$ in $\textsc{Construct}(r)$, then either it output $y$ in its last call to $\textsc{Construct}(r-1)$ or $p_{i,r}$ output $y$ in $q_i$'s local simulation of $p_{i,r}$ following this call. In either case, the last operation $q_i$ applied was in its last call to $\textsc{Construct}(r-1)$. Thus, by the induction hypothesis, the last operation $q_i$ applied was an $\real{M}.\id{Scan}$. Furthermore, if $q_i$ outputs $y$ in $\textsc{Construct}(r-1)$, then some process $p_{i,g}$, for $1 \leq g \leq r-1$, has output $y$. Hence, the second part of the claim holds. 
	
	Now suppose $q_i$ returns $([j_1,\dots,j_r],[v_1,\dots,v_r])$ from $\textsc{Construct}(r)$. This implies that $q_i$ did not terminate in any call to $\textsc{Construct}(r-1)$ during the call to $\textsc{Construct}(r)$. Since $\real{A}$ is initialized to empty immediately prior to the loop, by the code, it follows that $q_i$ calls $\textsc{Construct}(r-1)$ more than once.  From the code, it follows that the last call to $\textsc{Construct}(r-1)$ during $\textsc{Construct}(r)$ returned $([j_1,\dots,j_{r-1}],[v_1,\dots,v_{r-1}])$ and $\mathcal{A}$ contained some pair $(\{j_1,\dots,j_{r-1}\},\real{V}')$.
	By the code, this pair was added to $\real{A}$ when $q_i$ applied an atomic $\real{M}.\id{Block-Update}$ $\real{B}'_{r-1}$ to $\{j_1,\dots,j_{r-1}\}$. It applied this $\real{M}.\id{Block-Update}$ to simulate a block update returned by an earlier call to $\textsc{Construct}(r-1)$ made during $\textsc{Construct}(r)$. Following this call to $\textsc{Construct}(r-1)$, $q_i$ revises the past by locally simulating steps of $p_{i,r}$ assuming that the contents of $M$ are the same as $\real{V}'$. It does so until $p_{i,r}$ is poised to perform $M.\id{update}(j_r,v_r)$, for some $j_r \notin \{j_1,\dots,j_{r-1}\}$. Then $q_i$ returns $[(j_1,\dots,j_r),(v_1,\dots,v_r)]$. Thus, the last operation $q_i$ applied was in its last call to $\textsc{Construct}(r-1)$, which, by the induction hypothesis, was an $\real{M}.\id{Scan}$. Moreover, for $1 \leq g \leq r-1$, $p_{i,g}$ is poised to perform $M.\id{update}(j_g,v_g)$. Hence, the third part of the claim holds.
	
	Finally, suppose that $q_i$ does not terminate in any to call to $\textsc{Construct}(r-1)$ during the call to $\textsc{Construct}(r)$. Then, by the previous paragraph, we have shown the existence of $\real{B}'_{r-1}$. If $r = 2$, then $\real{B}'_{r-1} = \real{B}'_1$ and the fourth part of the claim holds. So suppose $r > 2$. By the induction hypothesis, in its last call to $\textsc{Construct}(r-1)$ during $\textsc{Construct}(r)$, $q_i$ applied atomic $\real{M}.\id{Block-Update}$s $\real{B}'_{r-2},\dots,\real{B}'_1$, in that order, such that, for each $1 \leq g \leq r-2$, $\real{B}'_g$ updates $g$ components, $q_i$ locally simulated steps $p_{i,g+1}$ assuming the contents of $M$ are the same as the view returned by $\real{B}'_g$, and, from $\real{B}'_g$ until the end of the procedure, $q_i$ does not apply any $\real{M}.\id{Block-Update}$s to $g+1$ components. Since $\real{B}'_{r-1}$ was applied before the last call to $\textsc{Construct}(r-1)$ began, it follows that $q_i$ applied $\real{B}'_{r-1}$ before $\real{B}'_{r-2}$. 
	
	In either case, by construction, $q_i$ only applies $\real{M}.\id{Block-Update}$s to at most $r-1$ components in $\textsc{Construct}(r)$. Hence, it does not apply any $\real{M}.\id{Block-Update}$s to $r$ components after $\real{B}'_{r-1}$ until $q_i$ terminates or returns from $\textsc{Construct}(r)$. Hence, the fourth part of the claim holds.
\end{proof}

We now prove the main properties of the covering simulator's algorithm.

\begin{lemma}
	\label{lem:covering-simulator}
	If $q_i$ is a covering simulator, then the following holds.
	\begin{enumerate}
		\item \label{lem:covering-simulator-alternating} $q_i$ alternately applies $\real{M}.\id{Scan}$ and $\real{M}.\id{Block-Update}$, until it applies an $\real{M}.\id{Scan}$ that causes it to terminate. 
		\item \label{lem:covering-simulator-operations} Each $\real{M}.\id{Scan}$ applied by $q_i$ simulates an $M.\id{scan}$ by $p_{i,1}$ and each $\real{M}.\id{Block-Update}$  applied by $q_i$ to $r$ components simulates $M.\id{update}$s by $p_{i,1},\dots,p_{i,r}$. Moreover, if an $\real{M}.\id{Block-Update}$ simulates an $M.\id{update}$ by $p_{i,r}$ then it updates at least $r$ components.
		\item \label{lem:covering-simulator-revision} If $q_i$ revises the past of process $p_{i,r}$, for some $r \geq 2$, immediately after applying an operation $\real{\delta}$, then the following holds.
		\begin{enumerate}
			\item \label{lem:covering-simulator-revision-return} $\real{\delta}$ is the last $\real{M}.\id{Scan}$ in a call, $\textsc{Q}$, to $\textsc{Construct}(r)$ such that $q_i$ returns from every call to $\textsc{Construct}(r-1)$ in $\textsc{Q}$.
			\item \label{lem:covering-simulator-revision-others} $q_i$ also revises the past of $p_{i,2},\dots,p_{i,r-1}$ immediately after $\real{\delta}$. 
			\item \label{lem:covering-simulator-revision-next} If $q_i$ does not terminate immediately after $\real{\delta}$, the next operation that $q_i$ applies is an $\real{M}.\id{Block-Update}$ to at least $r$ components. 
		\end{enumerate}
	\end{enumerate}
\end{lemma}
\begin{proof}
	By~\algorithmref{alg:coveringsimulator}, $q_i$ begins with a call $\textsc{Q}$ to $\textsc{Construct}(m)$. By~\propositionref{prop:construct-procedure}.\ref{prop:construct-procedure-altscanblock}, in $\textsc{Q}$, $q_i$ alternately performs $\real{M}.\id{Scan}$ and $\real{M}.\id{Block-Update}$. If $q_i$ terminates in $\textsc{Q}$, then, by~\propositionref{prop:construct-procedure}.\ref{prop:construct-procedure-terminate}, the last operation that it applies is $\real{M}.\id{Scan}$. Otherwise, after $q_i$ returns from $\textsc{Q}$, it only performs local computation. Hence, the last operation it applied was in $\textsc{Q}$, which,  by~\propositionref{prop:construct-procedure}.\ref{prop:construct-procedure-return}, is an $\real{M}.\id{Scan}$.
	
	By~\propositionref{prop:construct-procedure}.\ref{prop:construct-procedure-altscanblock}, each $\real{M}.\id{Scan}$ simulates an $\real{M}.\id{Scan}$ by $p_{i,1}$ and each $\real{M}.\id{Block-Update}$ to $r \leq m-1$ components simulates $M.\id{update}$s by $p_{i,1},\dots,p_{i,r}$. Since each $\real{M}.\id{Block-Update}$ $\real{B}$ simulates a block update $\beta$ returned by a call to $\textsc{Construct}(s)$ and $\textsc{Construct}(s)$ returns a block update by $p_{i,1},\dots,p_{i,s}$ (by~\propositionref{prop:construct-procedure}.\ref{prop:construct-procedure-return}), if $\real{B}$ simulates an $M.\id{update}$ by $p_{i,r}$, then $\beta$ contains an $M.\id{update}$ by $p_{i,r}$ and, hence, $s \geq r$.
	
	If $q_i$ revises the past of process $p_{i,r}$, then it must have done so in a (recursive) call $\textsc{Q}'$ to $\textsc{Construct}(r)$. Moreover, it did not terminate in any call to $\textsc{Construct}(r-1)$ in $\textsc{Q}'$. Thus, by~\ref{prop:construct-procedure}.\ref{prop:construct-procedure-blocks}, $\real{\delta}$ is an $\real{M}.\id{Scan}$ and $q_i$ also revised the pasts of $p_{i,2},\dots,p_{i,r-1}$ immediately after $\real{\delta}$. If $q_i$ does not terminate immediately after $\real{\delta}$, then it applies an $\real{M}.\id{Block-Update}$, which simulates the $M.\id{update}$ by $p_{i,r}$ and, hence, is to at least $r$ components. 
\end{proof}

The next proposition is useful in the step complexity analysis. In particular, an immediate consequence of the proposition is that the number of operations applied by a simulator $q_i$ is $2b+1$, where $b$ is the number of $\real{M}.\id{Block-Update}$s applied by $q_i$.

\begin{proposition}
	\label{prop:alternate-scan-blockupdate}
	Each simulator alternately applies $\real{M}.\id{Scan}$ and $\real{M}.\id{Block-Update}$, until it applies an $\real{M}.\id{Scan}$ that causes it to terminate. 
\end{proposition}
\begin{proof}
	Since a direct simulator directly simulates its process, using $\real{M}.\id{Block-Update}$ to simulate $M.\id{update}$ and $\real{M}.\id{Scan}$ to simulate $M.\id{scan}$, by~\assumptionref{ass:altscanupdate}, each direct simulator alternately applies $\real{M}.\id{Scan}$ and $\real{M}.\id{Block-Update}$, until it applies an $\real{M}.\id{Scan}$ that causes it to terminate. The claim for covering simulators follows by~\lemmaref{lem:covering-simulator}.\ref{lem:covering-simulator-alternating}.
\end{proof}

\subsection{The Intermediate Execution of a Real Execution}

Recall that $\real{M}$ is implemented from a single-writer atomic snapshot object. Hence, each step in a real execution is an operation on the underlying single-writer snapshot object. However, we proved that $\real{M}.\id{Scan}$s and the $\real{M}.\id{Update}$s that are part of $\real{M}.\id{Block-Update}$ operations are linearizable. Hence, for each real execution, we may consider its sequence
of linearized $\real{M}.\id{Scan}$ and $\real{M}.\id{Update}$ operations. By the simulation algorithm, each $\real{M}.\id{Scan}$ or $\real{M}.\id{Update}$ in this sequence simulates an $M.\id{scan}$ or $M.\id{update}$ by some process. In this section, we define the intermediate execution of a real execution, which facilitates the proof of correctness of the simulation.

To describe this execution, for each real process $q_i$, we specify how the states of the simulated processes in $P_i$ stored by $q_i$ and the contents of $\real{M}$ changes after each operation applied by $q_i$. In the linearized execution, $q_i$ applies $\real{M}.\id{Update}$s (instead of $\real{M}.\id{Block-Update}$s) and $\real{M}.\id{Scan}$s
and, crucially, $q_i$ updates the states of its simulated processes immediately after applying an operation (which may involve revising the pasts of some of its processes). Notice that, in a real execution, a simulator does not know when the $\real{M}.\id{Update}$s, which are part of some $\real{M}.\id{Block-Update}$ it has applied, are linearized. Hence, it only updates the states of its simulated processes after the $\real{M}.\id{Block-Update}$ completes.

%


More formally, for any real execution, we define an \emph{intermediate} execution, $\real{\sigma}$, as follows. Let $\real{\delta}_1,\real{\delta}_2,\dots$ be the sequence of linearized $\real{M}.\id{Scan}$ and $\real{M}.\id{Update}$ operations in the real execution, as described in~\sectionref{sec:fullaugsnapshot}. Then $\real{\sigma}$ is the sequence $\real{C}_0, \real{\delta}_1, \real{C}_1, \real{\delta}_2, \real{C}_2, \dots$, where each configuration
$\real{C}_k$ describes the contents of $\real{M}$ and the state of each simulated process. In particular, in the initial configuration $\real{C}_0$, the state of each simulated process $p_{i,g} \in P_i$ is the initial state of the process, with the same input as $q_i$. For each $k \geq 1$, the configuration $\real{C}_k$ is defined as follows.
\begin{itemize}
	\item Suppose $\real{\delta}_k$ is an $\real{M}.\id{Scan}$ by $q_i$. Then the contents of $\real{M}$ are the same at $\real{C}_{k-1}$ and $\real{C}_k$. The state of \emph{each} process $p_{i,g} \in P_i$ at $\real{C}_k$ is the state of $p_{i,g}$ stored by $q_i$ in the real execution after $\real{\delta}_k$ completes. The states of the other simulated processes are the same at $\real{C}_k$ and $\real{C}_{k-1}$.
	\item Suppose $\real{\delta}_k$ is an $\real{M}.\id{Update}(j,v)$ by $q_i$ that simulates an $M.\id{update}(j,v)$ by process $p_{i,g} \in P_i$. Then the contents of $\real{M}$ at $\real{C}_k$ are the same as at $\real{C}_{k-1}$, except component $j$ has value $v$. The state of $p_{i,g}$ at $\real{C}_k$ is the state of $p_{i,g}$ stored by $q_i$ in the real execution after the $\real{M}.\id{Block-Update}$ containing $\real{\delta}_k$ completes. The states of the other simulated processes are the same at $\real{C}_k$ and $\real{C}_{k-1}$.
\end{itemize}

Observe that the intermediate execution is neither an execution of the real system nor an execution of the simulated system. This is because the operations in the execution are applied by the simulators, while the configurations in the execution contain the states of the simulated processes. The next proposition shows that the intermediate execution behaves like an execution of the simulated system. In the next section, we show how to construct an actual execution of the system system from the intermediate execution by inserting the locally simulated steps. 

\begin{proposition}
	\label{obs:nextstep}
	For each $1 \leq i \leq f$, if $\real{\delta}$ is an operation in $\real{\sigma}$ by $q_i$ that simulates a step $\delta$ by a process $p_{i,g} \in P_i$, then $\delta$ is the next step of $p_{i,g}$ at the configuration in $\real{\sigma}$ immediately before $\real{\delta}$.
\end{proposition} 
\begin{proof}
	First consider $p_{i,1}$. If $q_i$ is a direct simulator, then, by~\algorithmref{alg:directsimulator}, each operation applied by $q_i$ simulates the next step of $p_{i,1}$ and, after the operation completes, $q_i$ updates the state of $p_{i,1}$. Observe that the same holds if $q_i$ is a covering simulator: by~\propositionref{prop:alternate-scan-blockupdate}, $q_i$ alternately applies $\real{M}.\id{Scan}$ and $\real{M}.\id{Block-Update}$ and, by~\propositionref{prop:construct-procedure}, each of these operations simulates a step of $p_{i,1}$. Since no steps of $p_{i,1}$ are locally simulated by $q_i$, each such operation by $q_i$, in fact, simulates the next step of $p_{i,1}$. Moreover, after each such operation completes, $q_i$ updates the state of $p_{i,1}$. By definition of $\real{\sigma}$, if no operation that simulates a step of $p_{i,1}$ has been applied, then $p_{i,1}$ is in its initial state. Otherwise, immediately after each operation applied by $q_i$ that simulates a step of $p_{i,1}$, the state of $p_{i,1}$ in $\real{\sigma}$ is the same as the state stored by $q_i$ in the real execution after the operation completes. This state remains unchanged until the next operation by $q_i$ that simulates a step of $p_{i,1}$. It follows that $\delta$ is the next step of $p_{i,1}$ at the configuration in $\real{\sigma}$ immediately before $\real{\delta}$.
	
	Now consider $p_{i,g}$, for $g > 1$. By~\propositionref{prop:construct-procedure}, each $\real{M}.\id{Scan}$ applied by $q_i$ simulates a $M.\id{scan}$ by $p_{i,1}$. Since $\real{\delta}$ simulates a step by $p_{i,g}$, for $g > 1$, $\real{\delta}$ is not an $\real{M}.\id{Scan}$. It follows that $\real{\delta}$ is an $\real{M}.\id{Update}$ that is part of some $\real{M}.\id{Block-Update}$ $\real{B}$ to $r \geq g$ components, which simulates a block update returned from a call to $\textsc{Construct}(r)$. By the third part of~\propositionref{prop:construct-procedure}, the last operation $q_i$ applied in this call to $\textsc{Construct}(r)$ is an $\real{M}.\id{Scan}$, $\real{\delta}'$. Moreover, $q_i$ updates the states of $p_{i,1},\dots,p_{i,r}$ immediately after $\real{\delta}'$ completes in the real execution so that they are poised to perform the $M.\id{update}$s simulated by $\real{B}$. Then, immediately after $\real{\delta}'$ in $\real{\sigma}$, $p_{i,g}$ is poised to perform $\delta$. By~\algorithmref{alg:coveringconstruction}, $\real{\delta}'$ is the last operation $q_i$ applied before $\real{B}$. Since no other process simulates steps by $p_{i,g}$, it follows that the state of $p_{i,g}$ at all configurations between $\real{\delta}'$ and $\real{\delta}$ is the same.
	Therefore, $\delta$ is the next step of $p_{i,g}$ immediately before $\real{\delta}$.
\end{proof}

By the properties of the augmented snapshot object $\real{M}$ as described in~\sectionref{sec:fullaugsnapshot}, the intermediate execution $\real{\sigma}$ has a special structure. In particular, the $\real{M}.\id{Update}$s that are part of an \emph{atomic} $\real{M}.\id{Block-Update}$ $\real{B}$ appear consecutively in $\real{\sigma}$. Furthermore, $\real{B}$ returns the contents of $\real{M}$ at a prefix $\real{\alpha}$ of $\real{\sigma}$ such that, between $\real{\alpha}$ and the first $\real{M}.\id{Update}$ that is part of $\real{B}$ in $\real{\sigma}$, there are no $\real{M}.\id{Scan}$s and no $\real{M}.\id{Update}$s that are part of other atomic $\real{M}.\id{Block-Update}$s. Let $\real{B}_1,\dots,\real{B}_\ell$ be the sequence of the \emph{completed} atomic $\real{M}.\id{Block-Update}$ operations in $\real{\sigma}$, i.e.~the $\real{M}.\id{Update}$s that are part of these $\real{M}.\id{Block-Update}$s all appear in $\real{\sigma}$. Then we may write the sequence of operations in $\real{\sigma}$ as $
\real{\alpha}_1\real{\gamma}_1\real{\beta}_1 \cdots \real{\alpha}_\ell\real{\gamma}_\ell\real{\beta}_\ell \real{\alpha}_{\ell+1}$,
where, for $1 \leq t \leq \ell$, $\real{B}_t$ returns the contents of $\real{M}$ at the configuration in $\real{\sigma}$ immediately after the last step in $\real{\alpha}_1\real{\gamma}_1\real{\beta}_1 \cdots \real{\alpha}_{t-1}\real{\gamma}_{t-1}\real{\beta}_{t-1}\real{\alpha}_t$, $\real{\gamma}_t$ contains only $\real{M}.\id{Update}$s that are part of non-atomic $\real{M}.\id{Block-Update}$s, $\real{\beta}_t$ is the sequence of $\real{M}.\id{Update}$s that comprise $\real{B}_t$ in $\real{\sigma}$ and $\real{\alpha}_{\ell+1}$ consists of all the operations in $\real{\sigma}$ following $\real{B}_{\ell}$. We call this the \emph{block decomposition} of $\real{\sigma}$.

\subsection{Correctness of the Simulation}
\label{sec:simproof}

In this section, we state and prove the main invariants of our simulation and use them to prove that our simulation solves the colorless task, $T$. Roughly, our invariants say that, for each intermediate execution (of a real execution), there is a corresponding (simulated) execution of the protocol $\Pi$ such that the state of each process $p_{i,g} \in P_i$ at the end of the simulated execution is the same as the state of $p_{i,g}$ at the end of the intermediate execution. By definition of the intermediate execution, this is the state of $p_{i,g}$ stored by $q_i$ at the end of the real execution, provided $q_i$ has no pending operation. The actual invariants are more complicated because we need to know the exact structure of the simulated execution in order to describe where the hidden steps of simulated processes are inserted.

\begin{lemma}
	\label{lem:geninv}
	Let $\real{\sigma}$ be the intermediate execution of a real execution from an initial configuration $\real{C}_0$ of the real system, let $\real{B}_1,\dots,\real{B}_\ell$ be the sequence of completed atomic $\real{M}.\id{Block-Update}$s in $\real{\sigma}$, and let  $\real{\alpha}_1\real{\gamma}_1\real{\beta}_1 \cdots \real{\alpha}_\ell\real{\gamma}_\ell\real{\beta}_\ell \real{\alpha}_{\ell+1}$ be the block decomposition of $\real{\sigma}$. Define $C_0$ to be the configuration of $\Pi$ in which, for $1 \leq i \leq f$ and $1 \leq g \leq |P_i|$, the input of each process $p_{i,g} \in P_i$ is the input of $q_i$. Then there is a possible execution $\sigma$ of the protocol, $\Pi$,  from an initial configuration $C_0$ of the simulated system, whose steps may be written as $\alpha_1\zeta_1\gamma_1 \beta_1 \cdots \alpha_\ell \zeta_\ell \gamma_\ell \beta_{\ell} \alpha_{\ell+1}$, such that:
	\begin{enumerate}
		\item \label{inv:correspondence} 
		\begin{enumerate} \item \label{inv:corr-prefix} For $1 \leq t \leq \ell$, $\alpha_t$, $\gamma_t$, and $\beta_t$ are obtained by replacing each operation in $\real{\alpha}_t$, $\real{\gamma}_t$, and $\real{\beta}_t$, respectively, with the step that it simulates. 
			\item \label{inv:corr-suffix} $\alpha_{\ell+1}$ is obtained by replacing each operation in $\real{\alpha}_{\ell+1}$ with the step that it simulates.
		\end{enumerate}
		
		\item \label{inv:states} For $1 \leq i \leq f$, the state of each process $p \in P_i$ at the end of $\simulated{\sigma}$ is the same as the state of $p$ at the end of $\real{\sigma}$.
		
		\item \label{inv:contents} For $1 \leq t \leq \ell+1$, if $\alpha_t'$ and $\real{\alpha}'_t$ are prefixes of $\alpha_t$ and $\real{\alpha}_t$, respectively, of the same length, then the contents of $\simulated{M}$ at configuration $\simulated{C}_0\simulated{\alpha}_1\zeta_1\gamma_1\beta_1  \cdots \alpha_{t-1}\zeta_{t-1}\gamma_{t-1}\beta_{t-1}\alpha_{t}'$ are the same as the contents of $\real{M}$ at configuration $\real{C}_0 \real{\alpha}_1 \real{\gamma}_1\real{\beta}_1 \cdots \real{\alpha}_{t-1} \real{\gamma}_{t-1}\real{\beta}_{t-1}\real{\alpha}_t'$.
		
		\item \label{inv:hiddensteps} For $1 \leq t \leq \ell$, if $\zeta_t$ is not empty, then the following properties hold:
		\begin{enumerate} 
			\item $\real{B}_t$ was applied by a covering simulator $q_i$. 
			\item After $\real{B}_t$, there is an $\real{M}.\id{Scan}$ $\real{\delta}'$ in $\real{\sigma}$ applied by $q_i$. 
			\item If $s$ is the number of components $\real{B}_t$ updates, then $\zeta_t$ is a solo execution by $p_{i,s+1}$. Moreover, immediately after $\real{\delta}'$, $q_i$ locally simulated $\zeta_t$ to revise the past of $p_{i,s+1}$ (using the view returned by $\real{B}_t$).
		\end{enumerate}  
	\end{enumerate}
\end{lemma}

\begin{proof}
	By induction on the length of $\real{\sigma}$. The base case is when $\real{\sigma} = \real{C}_0$. In this case, we define $\sigma = C_0$. Then $\sigma$ is a possible execution of $\Pi$ and property~\ref{inv:states} holds by definition of $C_0$. Property~\ref{inv:contents} holds since the contents $\real{M}$ and $\simulated{M}$ are initially the same. Properties~\ref{inv:correspondence} and \ref{inv:hiddensteps} are vacuously true. Therefore, $\sigma$ satisfies the invariant for $\real{\sigma}$. 
	
	Now suppose $\sigma$ satisfies the invariant for $\real{\sigma}$ and consider $\hat{\real{\sigma}}$, which contains an additional operation, $\real{\delta}$, by some simulator, $q_i$. Let $\delta$ be the step of $p_{i,g} \in P_i$ that is simulated by $\real{\delta}$ and let $\real{\alpha}_1\real{\gamma}_1\real{\beta}_1 \cdots \real{\alpha}_\ell \real{\gamma}_{\ell}\real{\beta}_\ell \real{\alpha}_{\ell+1}$ be the block decomposition of $\real{\sigma}$. By the induction hypothesis, the steps of $\sigma$ may be written as $\alpha_1\zeta_1\gamma_1\beta_1 \cdots \alpha_1\zeta_\ell\gamma_\ell\beta_\ell\alpha_{\ell+1}$ so that the invariant holds. We will define an execution $\hat{\sigma}$ that satisfies the invariant for $\hat{\real{\sigma}}$. 
	We consider two cases.
	
	\textbf{Case 1:} \emph{Immediately after $\real{\delta}$, $q_i$ does not revise the past of any process}. By property~\ref{inv:states} of the induction hypothesis, the state of $p_{i,g}$ at the end of $\sigma$ is the same as the state of $p_{i,g}$ at the end of $\real{\sigma}$. By~\propositionref{obs:nextstep}, $p_{i,g}$ is poised to perform $\delta$ immediately before $\real{\delta}$ in $\hat{\real{\sigma}}$, hence, at the end of $\real{\sigma}$. It follows that $\delta$ is the next step of $p_{i,g}$ at the end of $\sigma$. By property~\ref{inv:contents} of the induction hypothesis, the contents of $\real{M}$ at the end of $\real{\sigma}$ are the same as the contents of $\simulated{M}$ at the end of $\simulated{\sigma}$. If $\real{\delta}$ is a $\real{M}.\id{Scan}$, then this implies that $\real{\delta}$ and $\delta$ return the same output. Otherwise, $\real{\delta}$ and $\delta$ update the same component with the same value. Define $\hat{\sigma}$ to be the execution that is the same as $\sigma$, except it contains the additional step, $\delta$. Then the state of $p_{i,g}$ at the end of $\hat{\real{\sigma}}$ is the same as the state of $p_{i,g}$ at the end of $\hat{\simulated{\sigma}}$. The states of all other processes are unchanged. Hence, $\hat{\simulated{\sigma}}$ is a possible execution of $\Pi$ and property~\ref{inv:states} holds for $\hat{\sigma}$.
	
	To show that properties~\ref{inv:correspondence},~\ref{inv:contents}, and~\ref{inv:hiddensteps} hold, we consider the block decomposition of $\hat{\real{\sigma}}$. If $\real{\delta}$ is a $\real{M}.\id{Scan}$ or an $\real{M}.\id{Update}$ that is part of either a non-atomic $\real{M}.\id{Block-Update}$ or an $\real{M}.\id{Block-Update}$ that is incomplete in  $\hat{\real{\sigma}}$, then the block decomposition of $\hat{\real{\sigma}}$ is the same as that of $\real{\sigma}$, except with $\real{\delta}$ appended to the end of $\real{\alpha}_{\ell+1}$. Observe that we may write the steps of $\hat{\sigma}$ as $\alpha_1\zeta_1\gamma_1\beta_1 \cdots \alpha_\ell\zeta_\ell\gamma_\ell\beta_\ell\hat{\alpha}_{\ell+1}$, where $\hat{\alpha}_{\ell+1} = \alpha_{\ell+1}\delta$. Hence, property~\ref{inv:corr-suffix} holds. Properties~\ref{inv:corr-prefix}, \ref{inv:contents}, and \ref{inv:hiddensteps} are unaffected because they only refer to the parts of $\real{\sigma}$ and $\simulated{\sigma}$ up to and including $\real{\beta}_{\ell}$ and $\simulated{\beta}_{\ell}$, which are unchanged. 
	
	Now suppose that $\real{\delta}$ is the last $\real{M}.\id{Update}$ that is part a complete atomic $\real{M}.\id{Block-Update}$ $\real{B}$ in $\hat{\real{\sigma}}$. Then $\real{B} = \real{B}_{\ell+1}$ is the $(\ell+1)$'th such $\real{M}.\id{Block-Update}$ in $\hat{\real{\sigma}}$ and the block decomposition of $\hat{\real{\sigma}}$ is
	\[\real{\alpha}_1\real{\gamma}_1\real{\beta}_{\ell} \cdots \real{\alpha}_{\ell}\real{\gamma}_\ell\real{\beta}_\ell \real{\alpha}'_{\ell+1} \real{\gamma}_{\ell+1} \real{\beta}_{\ell+1} \real{\alpha}_{\ell+2} \, ,
	\]
	where $\real{\alpha}'_{\ell+1}$ is some prefix of $\real{\alpha}_{\ell+1}$, $\real{B}_{\ell+1}$ returns the contents of $\real{M}$ at $\real{C}_0 \real{\alpha}_1 \cdots \real{\alpha}_\ell\real{\gamma}_\ell\real{\beta}_\ell \real{\alpha}_{\ell+1}'$, $\real{\gamma}_{\ell+1}$ only contains $\real{M}.\id{Update}$s that are part of non-atomic $\real{M}.\id{Block-Update}$s, $\real{\beta}_{\ell+1}$ is the sequence of $\real{M}.\id{Update}$s that are part of $\real{B}_{\ell+1}$ in $\hat{\real{\sigma}}$, and $\real{\alpha}_{\ell+2}$ is empty.
	
	The steps of $\hat{\sigma}$ may be written as $\alpha_1\zeta_1\gamma_1 \beta_1 \cdots \alpha_\ell\zeta_\ell\gamma_\ell \beta_\ell \alpha'_{\ell+1}\zeta_{\ell+1}\gamma_{\ell+1}\beta_{\ell+1}\alpha_{\ell+2}$, where $\alpha'_{\ell+1}$, $\gamma_{\ell+1}$, and $\beta_{\ell+1}$ are obtained by replacing each operation in $\hat{\real{\alpha}}_{\ell+1}$, $\real{\gamma}_{\ell+1}$, and $\real{\beta}_{\ell+1}$ with the step that it simulates and both $\zeta_{\ell+1}$ and $\alpha_{\ell+2}$ are empty. Then property~\ref{inv:correspondence} holds by definition. Property~\ref{inv:hiddensteps} remains unchanged for $1 \leq t \leq \ell$ and holds for $t = \ell+1$ since $\zeta_{\ell+1}$ is empty. Finally, since property~\ref{inv:contents} holds for $\sigma$ and every prefix of $\real{\alpha}'_{\ell+1}$ is a prefix of $\real{\alpha}_{\ell+1}$, property~\ref{inv:contents} holds for $\hat{\sigma}$. 
	
	Therefore, $\hat{\sigma}$ satisfies the invariant for $\hat{\real{\sigma}}$ in this case.
	
	\textbf{Case 2:} \emph{Immediately after $\real{\delta}$, $q_i$ revises the past of some processes}. Since only a covering simulator may revise the past of its simulated processes, $q_i$ is a covering simulator. Consider the largest $r \geq 2$ such that $q_i$ revises the past of $p_{i,r}$ immediately after $\real{\delta}$. (Recall that $q_i$ never revises the past of $p_{i,1}$.) Then the states of processes $p_{i,r+1},\dots,p_{i,m}$ are unchanged after $\real{\delta}$. By~\lemmaref{lem:covering-simulator}.\ref{lem:covering-simulator-revision-return}, $\real{\delta}$ is the last $\real{M}.\id{Scan}$ applied by $q_i$ in a call, $\textsc{Q}$, to $\textsc{Construct}(r)$ such that $q_i$ returns from every call to $\textsc{Construct}(r-1)$ in $\textsc{Q}$. Hence,  by~\lemmaref{lem:covering-simulator}.\ref{lem:covering-simulator-operations}, $\real{\delta}$ simulates an $M.\id{scan}$, $\delta$, by $p_{i,1}$. Moreover, by~\lemmaref{lem:covering-simulator}.\ref{lem:covering-simulator-revision-others}, $q_i$ also revises the pasts of $p_{i,2},\dots,p_{i,r-1}$ immediately after $\real{\delta}$.
	
	Since $q_i$ returns from every call to $\textsc{Construct}(r-1)$ in $\textsc{Q}$, by~\propositionref{prop:construct-procedure}.\ref{prop:construct-procedure-blocks}, in $\textsc{Q}$, $q_i$ applied a sequence of atomic $\real{M}.\id{Block-Update}$s $\real{B}'_{{r-1}},\dots,\real{B}'_{1}$ such that, for each $1 \leq s \leq r-1$:
	\begin{enumerate} [(i)]
		\item \label{inv:case2-numcomps} $\real{B}'_{s}$ updates $s$ components,
		\item \label{inv:case2-revision} $q_i$ revised the past of $p_{i,s+1}$ using the view returned by $\real{B}'_{s}$ (i.e.~$q_i$ locally simulated a solo execution, $\xi_s$, of $p_{i,s+1}$, assuming the contents of $M$ are the same as the view returned by $\real{B}'_{s}$, and $\xi_s$ contains only $M.\id{update}$s to components updated by $\real{B}'_{s}$ and $M.\id{scan}$s), and
		\item \label{inv:case2-nomoresteps} from $\real{B}'_{s}$ until $q_i$ terminates or returns from $\textsc{Q}$ (i.e.~until immediately after $\real{\delta}$, since it is the last operation in $\textsc{Q}$), $q_i$ does not apply any $\real{M}.\id{Block-Update}$s to $s+1$ or more components. 
	\end{enumerate}  
	$\real{B}'_{r-1},\dots,\real{B}_1'$ are complete atomic $\real{M}.\id{Block-Update}$s in $\real{\sigma}$. Consider their indices $b_{r-1},\dots,b_1$ in the sequence of all complete atomic $\real{M}.\id{Block-Update}$s, $\real{B}_1,\dots,\real{B}_\ell$, in $\real{\sigma}$, i.e.~$\real{B}'_s = \real{B}_{b_s}$. Observe that, since $q_i$ applied $\real{B}'_{s}$ before it applied $\real{B}'_{s-1}$, $b_s < b_{s-1}$.

	Let $\tilde{\sigma}_0 = \sigma$. For $1 \leq g \leq r-1$, define $\tilde{\sigma}_g$ to be the revision of $\sigma$ where, for $r-g \leq h \leq r-1$, $\xi_h$ replaces $\zeta_{b_h}$ in $\sigma$. Intuitively, $\tilde{\sigma}_g$ contains the hidden steps of $p_{i,r},\dots,p_{i,r-g+1}$ that were locally simulated by $q_i$ immediately after $\real{\delta}$. We prove, by induction on $0 \leq g \leq r-1$, that $\hat{\sigma}_g = \tilde{\sigma}_g\delta$ is a valid execution of $\Pi$ that satisfies all properties of the invariant for $\hat{\real{\sigma}}$, except for property~\ref{inv:states}. Instead, it satisfies the following modified version of property~\ref{inv:states}:
	\begin{enumerate}
		\item[($2'$)] \label{inv:modifiedstates} The state of every process $p \in \{p_{i,2},\dots,p_{i,r-g}\}$ at the end of $\hat{\sigma}_g$ is the same as the state of $p$ at the end of $\real{\sigma}$ (i.e., its state has not yet been revised). The state of every other process $p$ at the end of $\hat{\sigma}_g$ is the same as the state of $p$ at the end of $\hat{\real{\sigma}}$.
	\end{enumerate}
	Recall that, if $q_i$ terminates after $\real{\delta}$ because it constructs a block update $\beta$ to $m$ components, then $q_i$ saves the states of $p_{i,1},\dots,p_{i,m}$ and, after it locally simulates $\beta$ and the solo execution of $p_{i,1}$, it restores their saved states. Thus, we do not need to treat this case any differently.
	
	The base case is $g = 0$. Note that $\real{\delta}$ is an $\real{M}.\id{Scan}$, so the block decomposition of $\hat{\real{\sigma}}$ is the same as that of $\real{\sigma}$, except with $\real{\delta}$ appended to the end of $\real{\alpha}_{\ell+1}$. 
	Observe that we may write the steps of $\hat{\sigma}_0$ as $\alpha_1\zeta_1\gamma_1\beta_1 \cdots \alpha_\ell\zeta_\ell\gamma_\ell\beta_\ell\hat{\alpha}_{\ell+1}$, where $\hat{\alpha}_{\ell+1} = \alpha_{\ell+1}\delta$. Hence, property~\ref{inv:corr-suffix} holds. Properties~\ref{inv:corr-prefix}, \ref{inv:contents}, and \ref{inv:hiddensteps} are unaffected because they only refer to the parts of $\real{\sigma}$ and $\simulated{\sigma}$ up to and including $\real{\beta}_{\ell}$ and $\simulated{\beta}_{\ell}$, which are unchanged. 
	
	Since $\sigma$ satisfies the invariant for $\real{\sigma}$, by property~\ref{inv:states}, the state of every process $p$ at the end of $\sigma$ is the same as the state of $p$ at the end of $\real{\sigma}$.  $\delta$ only changes the state of $p_{i,1}$.  Thus, the state of every other process is the same at the end of $\sigma$ and $\hat{\sigma}_0 = \sigma\delta$. 
	The state of $p_{i,1}$ is updated by $\real{\delta}$. By property~\ref{inv:contents} of the induction hypothesis for $\sigma$, the contents of $M$ at the end of $\sigma$ are the same as the contents of $\real{M}$ at the end of $\real{\sigma}$. Hence, $\real{\delta}$ and $\delta$ return the same view. It follows that the state of $p_{i,1}$ at the end of $\hat{\sigma}_0$ is the same as the state of $p_{i,1}$ at the end of $\hat{\real{\sigma}}$.
	Since $\hat{\sigma}_0$ does not contain the revisions of $p_{i,2},\dots,p_{i,r}$ caused by $\real{\delta}$, their states are the same at the end of $\hat{\sigma}_0$ and $\real{\sigma}$.
	The states of every other process does not change as a result of $\real{\delta}$, so its state at the end of $\real{\sigma}$ and $\hat{\real{\sigma}}$ are the same. Hence, its state at the end of
	$\hat{\sigma}_0$ is the same as its state at the end of 
	$\hat{\real{\sigma}}$.
	Thus, property $2'$ holds for $\hat{\sigma}_0$.
	
	Let $0 < g \leq r-1$ and suppose the claim holds for $g-1$. 
	Let $\hat{\sigma}_{g-1}'$ be the prefix of $\hat{\sigma}_{g-1}$ up to and including $\alpha_{b_{r-g}}$. Since $\hat{\sigma}_{g-1}$ is a valid execution of $\Pi$, $\hat{\sigma}'_{g-1}$ is a valid execution of $\Pi$. The remainder of $\hat{\sigma}_{g-1}$ is 
	\[
	\zeta_{b_{r-g}}\gamma_{b_{r-g}}\beta_{b_{r-g}} \alpha_{b_{r-g}+1}\zeta_{b_{r-g}+1}\gamma_{b_{r-g}+1}\beta_{b_{r-g}+1}\cdots \alpha_{\ell}\zeta_{\ell}\gamma_{\ell}\beta_{\ell}\alpha_{\ell+1}\delta \, .
	\]
	We first show that it does not contain any steps by $p_{i,r-g+1}$. We separately consider different parts of this suffix.
	
	Suppose, for a contradiction, that, for some $b_{r-g} \leq t \leq \ell$, $\zeta_t$ contains a step by $p_{i,r-g+1}$. Then, by property~\ref{inv:hiddensteps} of the induction hypothesis for $\sigma$, $\zeta_t$ is a solo execution by $p_{i,s+1}$, where $s$ is the number of components that $\real{B}_t$ updates. Moreover, after $\real{B}_t$, there is an $\real{M}.\id{Scan}$ $\real{\delta}'$ in $\real{\sigma}$ applied by $q_i$ such that, immediately after $\real{\delta}'$, $q_i$ locally simulates $\zeta_t$ to revise the past of $p_{i,s+1}$ using the view returned by $\real{B}_t$. Since $\zeta_t$ contains a step by $p_{i,r-g+1}$, $s = r-g$.  Observe that, since $\real{\delta}$ does not occur in $\real{\sigma}$, $\real{\delta}' \neq \real{\delta}$. Since $q_i$ applies $\real{\delta}$ after $\real{\delta}'$, it does not terminate immediately after $\real{\delta}'$. Thus, by~\lemmaref{lem:covering-simulator}.\ref{lem:covering-simulator-revision-next}, after $\real{\delta}'$, $q_i$ next applies an $\real{M}.\id{Block-Update}$ $\real{B}'$ to at least $s+1$ components. Since $\real{B}'$ occurs after $\real{B}_{b_s} = \real{B}'_s$ and before $\real{\delta}$, this contradicts (\ref{inv:case2-nomoresteps}).
	Thus, for $b_{r-g} \leq t \leq \ell$, $\zeta_t$ does not contain any steps by $p_{i,r-g+1}$.
	
	Suppose, for a contradiction, that, for some $b_{r-g}+1 \leq t \leq \ell$, $\alpha_t\gamma_t\beta_t$ contains a step by $p_{i,r-g+1}$. Then, by property~\ref{inv:correspondence} of the induction hypothesis for $\sigma$, $\real{\alpha}_t\real{\gamma}_t\real{\beta}_t$ contains an operation $\real{\delta}'$ that simulates a step by $p_{i,r-g+1}$. By~\lemmaref{lem:covering-simulator}.\ref{lem:covering-simulator-operations}, each $\real{M}.\id{Scan}$ applied by $q_i$ simulates an $M.\id{scan}$ by $p_{i,1}$. Hence, $\real{\delta}'$ is an $\real{M}.\id{Update}$ that is part of an $\real{M}.\id{Block-Update}$ to at least $r-g+1$ components. Since $\real{\alpha}_t\real{\gamma}_t\real{\beta}_t$ occurs after $\real{B}_{b_{r-g}} = \real{B}_{r-g}'$, this contradicts (\ref{inv:case2-nomoresteps}). 
	Thus, for $b_{r-g} + 1 \leq t \leq \ell$,  $\alpha_t\gamma_t\beta_t$ does not contain any steps by $p_{i,r-g+1}$. The same argument shows that $\alpha_{\ell+1}$ does not contain any steps by $p_{i,r-g+1}$. 
	
	By definition, $\real{\gamma}_{b_{r-g}}$ only contains $\real{M}.\id{Update}$s that are part of $\real{M}.\id{Block-Update}$s applied by other simulators. It follows, by property~\ref{inv:correspondence} of the induction hypothesis for $\sigma$, that $\gamma_{b_{r-g}}$ does not contain any steps by $p_{i,r-g+1}$. $\real{\beta}_{b_{r-g}}$ is an $\real{M}.\id{Block-Update}$ to $r-g$ components. Thus, by~\lemmaref{lem:covering-simulator}.\ref{lem:covering-simulator-operations} and property~\ref{inv:correspondence}, $\beta_{b_{r-g}}$ does not contain any steps by $p_{i,r-g+1}$. Finally, recall that $\delta$ is a step by $p_{i,1}$. Therefore, the remainder of $\hat{\sigma}_{g-1}$ after $\hat{\sigma}_{g-1}'$ does not contain any steps by $p_{i,r-g+1}$.
	
	\medskip
	
	
	We now show that $\xi_{r-g}$ is applicable at $C_0\hat{\sigma}_{g-1}'$. Since $p_{i,r-g+1}$ does not take steps in the suffix of $\hat{\sigma}_{g-1}$ following $\hat{\sigma}_{g-1}'$, the state of $p_{i,r-g+1}$ at $C_0\hat{\sigma}'_{g-1}$ is the same as the state of $p_{i,r-g+1}$ at $C_0\hat{\sigma}_{g-1}$. By property $2'$ of the induction hypothesis for $\hat{\sigma}_{g-1}$, the state of $p_{i,r-g+1}$ at the end of $\hat{\sigma}_{g-1}$ is the same as the state of $p_{i,r-g+1}$ at the end of $\real{\sigma}$. 
	By definition of how $q_i$ revises the past of $p_{i,r-g+1}$, the state of $p_{i,r-g+1}$ at the beginning of $\xi_{r-g}$ is the state of $p_{i,r-g+1}$ at the end of $\real{\sigma}$. By property~\ref{inv:contents} of the induction hypothesis for $\hat{\sigma}_{g-1}$, the contents of $M$ at $C_0\hat{\sigma}_{g-1}'$ is the same as the contents of $\real{M}$ at $\real{C}_0 \real{\alpha}_1\real{\gamma}_1\real{\beta}_1 \cdots \real{\alpha}_{b_{r-g}-1}\real{\gamma}_{b_{r-g}-1}\real{\beta}_{b_{r-g}-1}\real{\alpha}_{b_{r-g}}$, which is precisely the view of $\real{M}$ returned by $\real{B}_{b_{r-g}}$. By (\ref{inv:case2-revision}), it follows that $\xi_g$ is a valid solo execution of $p_{i,r-g+1}$ from $C_0\hat{\sigma}_{g-1}'$. 
	
	Finally, note that, since $\gamma_{b_{r-g}}\beta_{b_{r-g}}$ only contains $M.\id{update}$s, it is applicable after $\xi_{r-g}$. Moreover, since $\xi_{r-g}$ only contains $M.\id{updates}$ to components updated by $\beta_{b_{r-g}}$ (by (\ref{inv:case2-revision})), the contents of $M$ are the same at $C_0 \hat{\sigma}_{g-1}'\zeta_{b_{r-g}} \gamma_{b_{r-g}}\beta_{b_{r-g}}$ and $C_0 \hat{\sigma}_{g-1}'\xi_{r-g} \gamma_{b_{r-g}}\beta_{b_{r-g}}$.
	It follows that the remainder of $\hat{\sigma}_{g-1}$ after $\gamma_{b_{r-g}}\beta_{b_{r-g}}$ is applicable and the states of the other processes and the contents of $M$ do not change. It follows that properties~\ref{inv:correspondence} and~\ref{inv:contents} hold for $\hat{\sigma}_g$, since they hold for $\hat{\sigma}_{g-1}$. The state of $p_{i,r-g+1}$ is the same at the end of $\hat{\sigma}_g$ and $\hat{\sigma}$ since we have inserted the steps of $p_{i,r-g+1}$ locally simulated by $q_i$ immediately after $\real{\delta}$. The states of the other processes are unchanged. Thus, property $2'$ holds for $\hat{\sigma}_g$. Finally, property~\ref{inv:hiddensteps} holds for $t = b_{r-g}$ since $\real{\delta}$ appears at the end of $\hat{\real{\sigma}}$ and $\real{B}_{b_{r-g}}$ updates $r-g$ components (\ref{inv:case2-numcomps}). The other $\zeta_t$'s are unchanged. Therefore, property~\ref{inv:hiddensteps} holds for $\hat{\sigma}_g$.
	
	Observe that property $2'$ for $\hat{\sigma}_{r-1}$ is the same as property~\ref{inv:states}, so $\hat{\sigma} = \hat{\sigma}_{r-1}$ satisfies all properties of the invariant for $\hat{\real{\sigma}}$. Therefore, by induction, the claim holds for the entire execution.
\end{proof}

\begin{lemma}
	\label{lem:g1simvalid}
	The simulation solves the colorless task, $T$.
\end{lemma}
\begin{proof}
	Consider any real execution of the simulation from an initial configuration $\real{C}_0$. Let $\real{\sigma}$ be its intermediate execution. By~\lemmaref{lem:geninv}, there is a possible execution $\sigma$ of the protocol, $\Pi$, from an initial configuration $C_0$ of the simulated system that satisfies the invariants for $\real{\sigma}$. 
	
	Consider a covering simulator $q_i$ that returns from its call to $\textsc{Construct}(m)$. Recall that, after this call, $q_i$ locally simulates the block update $\beta$ returned by the call, followed by the terminating solo execution $\xi$ of $p_{i,1}$, and outputs the value that $p_{i,1}$ outputs in $\xi$. 
	
	Let $\bar{\sigma} = \sigma\beta\xi$. We claim that $\bar{\sigma}$ is also a possible execution of $\Pi$ from $C_0$. By~\lemmaref{lem:geninv}.\ref{inv:states}, the state of each simulated process $p_{i,g}$ at $C_0\sigma$ is the same as the state of $p_{i,g}$ stored by $q_i$ at $\real{C}_0\real{\sigma}$. Since $q_i$ restores the states of $p_{i,1},\dots,p_{i,m}$ after locally simulating $\beta\xi$, $p_{i,1},\dots,p_{i,m}$ are poised to perform $\beta$ at $C_0\sigma$. Thus, $\beta$ is applicable at $C_0\sigma$. Then, since $\beta$ overwrites the contents of all components of $M$, $\xi$ is applicable at $C_0\sigma\beta$. It follows that $\bar{\sigma}$ is a possible execution of $\Pi$ from $C_0$. 
	
	Since the covering simulators simulate disjoint sets of processes,  it is possible to append such executions from all the covering simulators onto the end of $\sigma$ and the resulting execution, $\bar{\bar{\sigma}}$ is still a possible execution of $\Pi$.
	
	Let $I = \{x_1,\dots,x_f\}$ be the set of inputs of the simulators in $\real{C}_0$ and let $O$ be the set of outputs in $\bar{\bar{\sigma}}$. Then $O$ is a valid output set for $I$, as specified by the task $T$. This is because, by construction, in $C_0$, each simulated process is assigned the input of its simulator, i.e.~the set of inputs of the simulated processes in $C_0$ is $I$ as well. Since $\Pi$ is assumed to be a correct protocol solving task $T$ and $\bar{\bar{\sigma}}$ is a possible execution of $\Pi$, $O$ is valid for $I$. Observe that, for each $1 \leq i \leq f$, exactly one process in $P_i$ has output a value in $\bar{\bar{\sigma}}$ and this is the value output by $q_i$. Thus, the set of outputs of the simulators is exactly $O$. It follows that the simulation is correct.
\end{proof}

\subsection{Wait-freedom of the simulation}

We now prove that the simulation is wait-free. We first prove lemmas that allow us to bound the number of operations that a covering simulator needs to apply.


\begin{proposition}
	\label{prop:numatomic}
	For $1 < r \leq m$, in any call to $\textsc{Construct}(r)$ by a covering simulator $q_i$, the size of the set $A$ on line 9 of~\algorithmref{alg:coveringconstruction} is at most ${m \choose r-1}$.
\end{proposition}
\begin{proof}
	Recall that $A$ contains pairs $(J,\real{V})$, where each $J$ is a set of $r-1$ components and $\real{V}$ is a view of $\real{M}$ returned by an atomic $\real{M}.\id{Block-Update}$ that $q_i$ applied during the call to $\textsc{Construct}(r)$. 
	Since $A$ is initially empty and a pair $(J,\real{V})$ is added to $A$ only if the test on line 12 is false, i.e.~only if $J$ is not equal to the first element of any pair in $A$, the first elements in $A$ are all distinct. Since each such element is a set of $r-1$ components and there ${m \choose r-1}$ different sets of components of size $r-1$, $|A| \leq {m \choose r-1}$.
\end{proof}

Let 
\begin{align*}
	a(r) &= 
	\begin{cases} 
		0 & \mbox{if }r= 1 \\
		\left({m \choose r-1}+1\right)a(r-1)+ {m \choose r-1} & \mbox{if }1 < r \leq m.  \\
	\end{cases} 
\end{align*}
It can be verified that $a(r) \leq ({m \choose m/2}+1)^{r-1} - 1 \leq 2^{m(r-1)}$. 


\begin{lemma}
	\label{lem:numatomic}
	For $1 \leq r \leq m$, if every $\real{M}.\id{Block-Update}$ applied by $q_i$ during a call to $\textsc{Construct}(r)$ is atomic, then the maximum number of  $\real{M}.\id{Block-Update}$s applied by $q_i$ in a call to $\textsc{Construct}(r)$ is at most $a(r)$.
\end{lemma}
\begin{proof}
	By induction on $r$. The base case is $r = 1$. It holds since $q_i$ applies $a(1) = 0$ $\real{M}.\id{Block-Update}$s in $\textsc{Construct}(1)$. Now let $r>1$ and suppose the claim holds for $r-1$. By~\propositionref{prop:numatomic}, in a call to $\textsc{Construct}(r)$ by $q_i$, the size of the set $\mathcal{A}$ on line 9 of~\algorithmref{alg:coveringconstruction} is at most ${m \choose r-1}$. By the test on line 20, a pair is added to $\mathcal{A}$ exactly when $q_i$ applies an atomic $\real{M}.\id{Block-Update}$ on line 19. By assumption, each $\real{M}.\id{Block-Update}$ applied by $q_i$ during this call is atomic. Hence, every $\real{M}.\id{Block-Update}$ applied by $q_i$ during 
	its recursive calls to $\textsc{Construct}(r-1)$ are atomic. In each such call, $q_i$ applies at most $a(r-1)$ $\real{M}.\id{Block-Update}$s, by the induction hypothesis.
	Excluding these $\real{M}.\id{Block-Update}$s 
	$q_i$ applies at most ${m \choose r-1}$  $\real{M}.\id{Block-Update}$s in its call to $\textsc{Construct}(r)$. 
	Each of these $\real{M}.\id{Block-Update}$s is immediately preceded by a recursive call by $q_i$ to $\textsc{Construct}(r-1)$.
	Furthermore, after applying the last $\real{M}.\id{Block-Update}$, $q_i$ calls $\textsc{Construct}(r-1)$ once more. It follows that $q_i$ applies at most 
	$({m \choose r-1}+1)a(r-1) + {m \choose r-1} = a(r)$ atomic $\real{M}.\id{Block-Update}$s in $\textsc{Construct}(r)$.
\end{proof}

Let 
\[
b(i) = 
\begin{cases}
a(m) & \mbox{if }i = 1 \\
(a(m-1)+1)\sum_{j=1}^i b(j) + a(m) & \mbox{if } 1 < i \leq f .
\end{cases}
\]
It can be verified that $b(i) = a(m)(a(m-1)+1)^{i-1} \leq a(m)^i \leq 2^{i m(m-1)}$ for $1 \leq i \leq f$.

\begin{lemma}
	\label{lem:numblocks}
	For $1 \leq i \leq f-x$, the maximum number of  $\real{M}.\id{Block-Update}$s that covering simulator $q_i$ applies in any real execution is at most $b(i)$.
\end{lemma}
\begin{proof}
	By induction on $i$. The base case is $i=1$. Since $q_1$ has the smallest identifier, all of its $\real{M}.\id{Block-Update}$s are atomic. Hence, by~\lemmaref{lem:numatomic}, $q_1$ applies at most $a(m) = b(1)$ $\real{M}.\id{Block-Update}$s. 
	Now let $i > 1$ and suppose that the claim holds for $i-1$.
	In this case, we also need to count the $\real{M}.\id{Block-Update}$s applied by $q_i$ that are not atomic. By property X of an augmented snapshot object [make a reference], if an $\real{M}.\id{Block-Update}$  $\real{B}$ applied by $q_i$ is not atomic, then some covering simulator $q_j$ with $j < i$ applied an $\real{M}.\id{Block-Update}$ in the execution interval of $\real{B}$. 
	
	By the induction hypothesis, $q_1,\dots,q_{i-1}$ collectively apply at most $\sum_{j=1}^{i-1} b(j)$ $\real{M}.\id{Block-Update}$s in total during the execution. In the worst case, each $\real{M}.\id{Block-Update}$ applied by these simulators causes a different $\real{M}.\id{Block-Update}$ applied by $q_i$ to be non-atomic. The non-atomic $\real{M}.\id{Block-Update}$s applied by $q_i$ and the $\real{M}.\id{Block-Update}$s that $q_i$ applied to construct them are all useless
	for constructing an atomic $\real{M}.\id{Block-Update}$ to $m$ components. Since $a(r) > a(r-1)$ for all $1 < r \leq m$, $q_i$ applies the maximum number of $\real{M}.\id{Block-Update}$s when only $\real{M}.\id{Block-Update}$s applied by $q_i$ to $m-1$ components are non-atomic and, hence, all $\real{M}.\id{Block-Update}$s applied by $q_i$ in its calls to $\textsc{Construct}(m-1)$ are atomic. 
	
	By~\propositionref{prop:numatomic}, $q_i$ applies at most $m$ atomic $\real{M}.\id{Block-Update}$s to $m-1$ components. Since all $\real{M}.\id{Block-Update}$s applied by $q_i$ in its calls to $\textsc{Construct}(m-1)$ are atomic, by~\lemmaref{lem:numatomic}, it applies at most $a(m-1)$ $\real{M}.\id{Block-Update}$s in each such call. Each atomic $\real{M}.\id{Block-Update}$ is immediately preceded by a call to $\textsc{Construct}(m-1)$. Furthermore, after applying the last $\real{M}.\id{Block-Update}$, $q_i$ calls $\textsc{Construct}(m-1)$ once more. 
	
	Therefore, in total, $q_i$ applies at most $(a(m-1)+1)\sum_{j=1}^{i-1} b(j) + (m+1) a(m-1) + m = b(i)$ $\real{M}.\id{Block-Update}$s.
\end{proof}

\begin{lemma}
	\label{lem:finitecovering}
	Each covering simulator $q_i$ applies at most $2b(i)+1$ operations.
	Moreover, if there are only covering simulators, i.e.~$x = 0$, then every covering simulator outputs a value after taking at most $(2f+7)b(f)+3 \leq 2^{fm^2}$ steps.
\end{lemma}
\begin{proof}
	By~\propositionref{prop:alternate-scan-blockupdate}, each simulator alternately applies $\real{M}.\id{Scan}$ and $\real{M}.\id{Block-Update}$ until it applies an $\real{M}.\id{Scan}$ that causes it to terminate. By~\lemmaref{lem:numblocks}, each covering simulator $q_i$ applies at most $b(i)$ $\real{M}.\id{Block-Update}$s and, hence, at most $b(i)+1$ $\real{M}.\id{Scan}$s. It follows that $q_i$ applies at most $2b(i)+1$ operations in total. 
	
	By~\lemmaref{lem:impldetails}, each $\real{M}.\id{Block-Update}$ operation consists of 6 steps and each $\real{M}.\id{Scan}$ operation $S$ consists of at most $2k_S+3$ steps, where $k_S$ is the number of different updates by other simulators that are concurrent with it. Notice that $\sum [ k_S : S\textrm{ is a }\real{M}.\id{Scan}\textrm{ by }q_i]$ is bounded above by the number of $\real{M}.\id{Block-Update}$ operations applied by the other simulators.
	
	If there are only covering simulators, then $\sum [ k_S : S\textrm{ is an }\real{M}.\id{Scan}\textrm{ by }q_i] \leq \sum_{j \neq i} b(j) \leq (f-1)b(f)$. In each of its at most $b(f)$ $\real{M}.\id{Block-Update}$ operations, $q_i$ takes $6$ steps. Moreover, $q_i$ takes at most $\sum [ 2k_S + 3 : S\textrm{ is an }\real{M}.\id{Scan}\textrm{ by }q_i] \leq  2(f-1)b(f) + 3(b(i)+1)$ steps in its $\real{M}.\id{Scan}$ operations. Therefore, $q_i$ takes at most $6b(f) + 2(f-1)b(f) + 3(b(i)+1) \leq (2f+7)b(f) + 3$ steps in total. Since $(2f+7)b(f) + 3 \leq (2f+7)2^{fm(m-1)} + 3 \leq 2^{fm^2}$ whenever $f \geq 2$ and $m \geq 2$, the desired bound follows.
\end{proof}

\begin{lemma}
	\label{lem:g1simterm}
	The simulation is wait-free.
\end{lemma}
\begin{proof}
	\lemmaref{lem:finitecovering} takes care of the case when there are $x = 0$ direct simulators, so consider $x > 0$. Recall that, in this case, we assume that $\Pi$ is $x$-obstruction-free.
	Suppose, for a contradiction, that there is a real execution where 
	some process $q_i$ applies infinitely many operations on $\real{M}$. 
	Let $\real{\sigma}$ be its intermediate execution. 
	By~\lemmaref{lem:finitecovering}, covering simulators  
	apply only finitely many operations in $\real{\sigma}$, so $q_i$ is a direct simulator. Since $\real{\sigma}$ is infinite, there is an infinite suffix $\real{\sigma}''$ of $\real{\sigma}$ in which only direct simulators apply operations. Let $\real{\sigma}'$ be the prefix of $\real{\sigma}$ prior to $\real{\sigma}''$ let $\sigma'$ be a simulated execution of $\Pi$ that satisfies~\lemmaref{lem:geninv} for $\real{\sigma}'$. Let $\sigma''$ be the execution obtained by replacing each operation in $\real{\sigma}''$ with the operation that it simulates.
	
	By~\lemmaref{lem:geninv}.\ref{inv:states} for $\sigma'$, the state of each process $p_{i,g}$ at the end of $\sigma'$ is the same as the state of $p_{i,g}$ at the end of $\real{\sigma}'$. Moreover, by~\lemmaref{lem:geninv}.\ref{inv:contents} for $\sigma'$, the contents of $M$ and $\real{M}$ are the same at the end of $\sigma'$ and $\real{\sigma}'$, respectively. Thus, $\sigma''$ is applicable at the end of $\sigma'$. Only the $x$ direct simulators take steps in $\sigma''$. Since $\sigma''$ is an infinite execution, this contradicts the fact that $\Pi$ is $x$-obstruction-free. Thus, every simulator $q_i$ applies only finitely many operations on $\real{M}$.
	
	Our implementation of $\real{M}$ is non-blocking. So, if processes perform infinitely many accesses 
	to the underlying single-writer snapshot object in the implementation, 
	then infinitely many operations on $\real{M}$ will complete. 
	Since every process applies only finitely many operations on $\real{M}$, 
	there is no infinite execution, which means that the simulation is wait-free.
\end{proof}

\subsection{Applications of the Simulation Theorem}

The following bounds are immediate corollaries of the simulation theorem.

\begin{corollary}
	For $1 \leq x \leq k$, any $x$-obstruction-free protocol for solving $k$-set agreement among $n \geq k+1$ processes uses at least $\lfloor \tfrac{n-x}{k+1-x} \rfloor +1$ registers. In particular, any obstruction-free protocol for solving consensus among $n \geq 2$ processes uses at least $n$ registers.
\end{corollary}
\begin{proof}
	It is known that it is impossible to solve $k$-set agreement among $k+1$ processes in a wait-free manner~\cite{BG93,HS99,SZ00,AC13,AP16}. Thus, the desired bound follows by applying the second part of~\theoremref{thm:simulation} with $f = k+1$.
\end{proof}

\begin{corollary}
	For $0 < \epsilon < 1$, any obstruction-free protocol for solving $\epsilon$-approximate agreement among $n \geq 2$ processes uses at least $\min \{\lfloor \tfrac{n}{2} \rfloor +1, \sqrt{\log_2 \log_3(\tfrac{1}{\epsilon}) - 2} \}$ registers.
\end{corollary}
\begin{proof}
	It is known that any protocol for solving $\epsilon$-approximate agreement among 2 processes takes at least $L = \frac{1}{2}\log_3(\frac{1}{\epsilon})$ steps~\cite{HS2006}. Thus, desired bound follows by applying the first part of~\theoremref{thm:simulation} with $f = 2$.
\end{proof}
\section{Nondeterministic Solo Terminating to Obstruction-free}
\label{sec:ndst2of}

A protocol is \emph{nondeterministic solo terminating} if, for every process $p$ and every configuration $C$, there exists a solo execution by $p$ from $C$ in which $p$ outputs a value (and terminates). 
This property was introduced by Ellen, Herlihy, and Shavit to prove space lower bounds for randomized wait-free and obstruction-free consensus protocols.
In this section, we prove that, for a large class of objects, the space complexity of
obstruction-free protocols and 
nondeterministic solo terminating protocols
using only these objects
is the same.

\subsection{Protocols}

A \emph{nondeterministic protocol} specifies a nondeterministic state machine $M_p$ for each process $p$. Each \emph{state machine} $M_p$ is a 5-tuple 
$(S_p,\nu_p,\delta_p,I_{p},F_p)$,
where
\begin{itemize}
	\item
	$S_p$ is a totally ordered set of states,
	\item
	$I_p \subseteq S_p$ is a set of initial states, one for each possible input to $p$,
	\item
	$F_p \subseteq S_p$ is a set of final states, one for each possible output of $p$,
	\item
	$\nu_p$ specifies the next step that $p$ will perform in each non-final state $s \in S_p - F_p$,
	and
	\item
	$\delta_p$ is a transition function mapping each non-final state
	$s \in S_p - F_p$
	and possible response from step $\nu_p(s)$ to a nonempty subset of $S_p$. 
\end{itemize}

A state machine is \emph{deterministic}
if $\delta_p$ maps each non-final state and possible response to a single state, rather than a subset of states. A protocol is \emph{deterministic} if the state machine of every process is deterministic.

In every configuration, each process $p$ is in some state $s \in S_p$.
Initially, each process is in one of its initial states and each component $j$ of the snapshot contains an initial value, $v_j$.
When allocated a step by the scheduler, 
$p$ does nothing if it is in a final state.
If $p$ is in  state $s \in S_p - F_p$, it performs step $\nu_p(s)$.
If $a$ is the response it receives from this step,
$p$ chooses its next state $s'$ from $\delta_p(s,a)$.
If $s'$ is in $F_p$, we consider $p$ to have output the value corresponding to $s'$.

\subsection{$m$-component objects}

An \emph{$m$-component object} supports a \id{scan} operation, which returns the values
of all $m$ components, and a set of operations on individual components. For example, in addition to \id{scan},
an $m$-component snapshot object
supports \id{write} to each component.
An $m$-component max register
supports \id{scan} and \id{writemax} to each component.

\begin{theorem}
	\label{thm:ndst2of}
	If there is a nondeterministic solo terminating protocol for a task that only uses one $m$-component object, then there is an obstruction-free protocol for the task that only uses the same object.
\end{theorem}
\begin{proof}
	Let $\Pi$ be a nondeterministic solo terminating protocol for a task $T$ that uses a single $m$-component object $M$. Let $M_p = (S_p,\nu_p,\delta_p,I_p,F_p)$ be the state machine specified by $\Pi$ of each process $p$. 
	
	Without loss of generality, we may assume that, in $\Pi$, each process alternately performs \id{scan} and operations on components of $M$, starting with a \id{scan}, until it performs a \id{scan} that causes it to output a value and terminate. Moreover, we may assume that each process $p$ stores a vector $E_p$, as a part of its state, which is updated as follows. Initially, $E_p = (v_1,\dots,v_m)$, where each $v_j$ is the initial value of component $j$ of $M$.
	Whenever $p$ performs a \id{scan} on $M$, it updates $E_p$ to the result of the \id{scan}. After performing an operation on component $j$ of $M$, $p$ updates $E_p$ by simulating the operation on component $j$ of $E_p$. At all times, $E_p$ contains what $p$ expects to see when it next performs a \id{scan}, provided no other process has taken any steps since its last \id{scan} of $M$. For each state $s \in S_p$, let $E_p(s)$ denote the contents of $E_p$ stored by $p$ in state $s$. 
	
	\medskip 
	A \emph{$p$-solo path} from a state $s$ of length $t$ is an alternating sequence of states 
	in $S_p$ and responses, $s_0, a_0, s_1, \dots, a_{t-1}, s_t$, such that $s_0 = s$, 
	$s_t \in F_p$, and, for each $0 \leq i < t$,
	\begin{itemize}
		\item 
		$s_{i+1} \in \delta_p(s_i,a_i)$,
		\item 
		if $\nu_p(s_i)$ is an operation
		on component $j$, then $a_i$ is the
		response when
		$p$ simulates this operation on component $j$ of $E_p(s_i)$, and
		\item
		if $\nu_p(s_i)$ is a \id{scan}, then $a_i = E_p(s_i)$.
	\end{itemize}
	
	A $p$-solo path from $s$ represents a solo terminating execution by $p$ from a configuration $C$ where $p$ is in state $s$ and the contents of $M$ are $E_p(s)$. By nondeterministic solo termination, if $C$ is reachable, then there is a terminating solo execution by $p$ from $C$. Hence, a $p$-solo path from $s$ exists.
	In particular, if $C$ is a reachable configuration immediately following a \id{scan} by $p$
	and $p$ is in state $s$ in $C$, then the contents of $M$ in
	$C$ are $E_p(s)$. Hence, there is a $p$-solo path from $s$.
	However, if $C$ is not reachable, then there is not necessarily a solo execution by $p$ from $C$ that
	reaches a final state.
	
	
	\medskip
	
	For each process $p$, we define a deterministic state machine $M_p' = (S_p,\nu_p,\delta'_p,I_{p},F_p)$ from $M_p = (S_p,\nu_p,\delta_p,I_{p},F_p)$.
	The result is a deterministic protocol, $\Pi'$, which uses the same $m$-component object.
	For each state $s \in S_p - F_p$ and response $a$, we define $\delta_p'(s,a)$ as follows.
	If there is a $p$-solo path from $s$ such that the response of step $\nu_p(s)$ is $a$,
	then consider the first state $s'$
	such that there is a shortest
	$p$-solo path that begins with $s,a,s'$
	and define $\delta'_p(s,a) = s'$. 
	Otherwise, consider the first state $s' \in \delta_p(s,a)$ and define $\delta'_p(s,a) = s'$. 
	
	\medskip

	For each process $p$, each state $s \in S_p - F_p$, and each response $a$ to $\nu_p(s)$,
	$\delta'_p(s,a) \in \delta_p(s,a)$.
	Thus $\delta'_p(s,a)$ is a state that $p$ could be in after performing $\nu_p(s)$ in $\Pi$.
	Hence every execution of $\Pi'$ is an execution of $\Pi$.  It follows that if $\Pi$ is a  protocol for the task
	$T$, $\Pi'$ is also a protocol for $T$.
	
	\medskip
	
	Assume, for a contradiction, that there is an infinite solo execution $\alpha$ by some process $p$ from a reachable configuration $C$ of $\Pi'$. In this execution, $p$ never enters a final state and, so, does not output a value. 
	Let $\alpha'$ be the shortest prefix of $\alpha$ that contains a \id{scan} by $p$ and write $\alpha = \alpha'\alpha''$. Let $C_0 = C\alpha'$, $e_0$, $C_1$, $e_1, \dots$ be the alternating sequence of configurations and steps of $\alpha''$. For $i \geq 0$, let $s_i \in S_p - F_p$ be the state of $p$ in $C_i$ and let $a_i$ be the response of $e_i$. 
	
	Since $\alpha'$ ends with a \id{scan} by $p$, the contents of $M$ are equal to $E_p(s_0)$ at $C_0$. Moreover, since $\alpha''$ is a solo execution by $p$, the contents of $M$ are equal to $E_p(s_i)$ at $C_i$, for each $i \geq 0$.
	Hence, there is a $p$-solo path from $s_i$, for all $i \geq 0$.
	
	If $e_i$ is a \id{scan}, then its response, $a_i$, is exactly $E_p(s_i)$. Similarly, if $e_i$ is an operation on a component, its response, $a_i$, is equal to the response $p$ obtains when it simulates this operation on $E_p(s_i)$. Thus, in any  $p$-solo path from $s_i$, the response returned by $\nu_p(s_i) = e_i$ is equal to  $a_i$. 
	
	For $i \geq 0$, let $\ell_i$ be the length of a shortest $p$-solo path from $s_i$. 
	By definition of $\delta'_p(s_i,a_i)$, $s_{i+1}$ is the first state such that there is a shortest $p$-solo path $\sigma_i$ from $s_i$ that begins with $s_i,a_i,s_{i+1}$. Since the suffix of $\sigma_i$ starting from $s_{i+1}$ is also a $p$-solo path from $s_{i+1}$, it follows that $\ell_{i+1} = \ell_i - 1$. 
	
	This implies that, when
	$i = \ell_0$,
	$\ell_{i} = 0$ and, hence, $s_{i} \in F_p$. This contradicts the fact that $s_i \in S_p - F_p$. Hence, $\Pi'$ is obstruction-free.
\end{proof}

\subsection{General Objects}

Consider a nondeterministic solo terminating protocol $\Pi$ that uses $m$ objects, each of which supports a \id{read} operation. We say that $\Pi$ is \emph{ABA-free} if, for any execution $C_0, e_0, C_1, e_1, \dots$ of $\Pi$, there is no $i < j < k$ such that the value of component $c$ at configurations $C_i$ and $C_k$ are the same, but the value of component $c$ at configuration $C_j$ is different.

Observe that if $\Pi$ uses only fetch-and-increment objects or max-registers, then $\Pi$ is ABA-free. If $\Pi$ uses only registers, then $\Pi$ can be made ABA-free by having each process append its identifier and a strictly increasing sequence number to each of its write operations. These extra values are ignored by reads. 
Similarly if $\Pi$ uses only
swap objects or compare-and-swap objects, it can be made ABA-free.

\begin{corollary}
	Suppose there is a non-deterministic solo terminating protocol $\Pi$ that only uses $m$ objects, $r_1,\dots,r_m$, each of which supports read. If $\Pi$ is ABA-free, then there is an obstruction-free protocol that uses the same $m$ objects.
\end{corollary}
\begin{proof}
	Let $O_j$ denote the set of operations supported by $r_j$.
	It is possible to simulate $\Pi$ using an $m$-component object such that component $j$ of the object supports the same operations, $O_j$. Hence, there is a nondeterministic solo terminating protocol $\Pi'$ that uses an $m$-component object. 
	By Theorem~\ref{thm:ndst2of}, there is an obstruction-free protocol $\Pi''$ that uses the same $m$-component object. Since $\Pi$ is ABA-free, both $\Pi'$ and $\Pi''$ are ABA-free. 
	Observe that we can simulate $\Pi''$ using $r_1,\dots,r_m$. In particular, we may simulate a scan with obstruction-free double collects. Since $\Pi''$ is ABA-free, this is linearizable. We can simulate an operation in $O_j$ on component $j$ with the same operation on $r_j$. The result is an obstruction-free protocol $\Pi'''$ that uses $r_1,\dots,r_m$.
\end{proof}

\section{Conclusions and Future Work}
We conjecture that the space complexity of $x$-obstruction-free $k$-set agreement 
is $n-k+x$, matching the upper bound of~\cite{BRS15}. 
Our paper makes significant progress by proving the first 
  non-constant lower bound for non-anonymous processes.
This lower bound is asymptotically tight when
  $k$ and $x$ are constant and tight 
when $x=1$ and either $k=1$ (obstruction-free consensus) and
  $k=n-1$ (obstruction-free $(n-1)$-set agreement).
Equally importantly, our simulation technique uses a
new
approach.
It is conceivable that 
this approach
can be 
extended to obtain a tight lower bound
when $k > 1$ and
$1 \leq x \leq k$.


We proved that a space lower bound for obstruction-free protocols implies
a
space lower bound for protocols that satisfy nondeterministic solo termination
(including randomized wait-free protocols)
by converting any nondeterministic solo terminating protocol to an obstruction-free
protocol that uses the same number of registers.
This allows researchers to focus on deriving space lower bounds 
  for obstruction-free protocols.

\nopagebreak
It is known how to convert any deterministic obstruction-free protocol for $n$ processes
that has solo step complexity $b$ into a randomized wait-free protocol against  an oblivious adversary, which has expected step complexity polynomial in $n$ and $b$~\cite{GHHW13}.
However, our construction provides no bound on the solo step complexity of the resulting obstruction-free protocol.
It would be interesting to improve the construction to bound the solo
step complexity of the resulting protocols.

%



\bibliographystyle{plainurl}
\bibliography{biblio}



\end{document}